\documentclass[10pt, Conference, letterpaper]{IEEEtran}
\IEEEoverridecommandlockouts
\usepackage{cite}
\usepackage{lipsum}
\usepackage[most]{tcolorbox}
\usepackage{booktabs}
\usepackage[T1]{fontenc}
\usepackage{pifont}
\usepackage{aecompl}
\usepackage[numbers,sort&compress]{natbib}
\usepackage{amsmath,amssymb,amsfonts}
\usepackage{amsthm}
\usepackage{float}
\usepackage{booktabs}
\makeatletter
\renewcommand{\maketag@@@}[1]{\hbox{\m@th\normalsize\normalfont#1}}%
\makeatother
\usepackage{graphicx}
\usepackage{verbatim}
\usepackage{subcaption}
\usepackage{textcomp}
\usepackage{xcolor}
\usepackage{stfloats}
\usepackage{tabularx}
\usepackage{array}
\usepackage{makecell}
\usepackage{url}
\makeatother
\usepackage{textcomp}
\usepackage[marginal]{footmisc}
\usepackage{geometry}
\usepackage[hidelinks]{hyperref}
\usepackage{setspace}
\usepackage{caption}
\allowdisplaybreaks[4]
\usepackage{algpseudocode}
\usepackage[ruled,linesnumbered]{algorithm2e}

\def\BibTeX{{\rm B\kern-.05em{\sc i\kern-.025em b}\kern-.08em
    T\kern-.1667em\lower.7ex\hbox{E}\kern-.125emX}}

\newtheorem{myTheo}{Theorem}

\begin{document}
\title{AoI-Aware Task Offloading and Transmission Optimization for Industrial IoT Networks: A Branching Deep Reinforcement Learning Approach}
\author{Yuang Chen, \IEEEmembership{Graduate Student Member,~IEEE,} Fengqian Guo, Chang Wu, Shuyi Liu, \\Hancheng Lu, \IEEEmembership{Senior Member,~IEEE,} Chang Wen Chen, \IEEEmembership{Life Fellow,~IEEE}
\thanks{This work was supported in part by the National Science Foundation of China under Grant U21A20452, in part by Hong Kong Research Grants Council (GRF-15213322, GRF-15229423). Yuang Chen, Chang Wu, Shuyi Liu, Fengqian Guo, and Hancheng Lu are with the Laboratory of Future Networks, University of Science and Technology of China (USTC), Hefei, P. R. China (e-mail: \{yuangchen21, changwu, lsy576083832\}@mail.ustc.edu.cn; fqguo@ustc.edu.cn; hclu@ustc.edu.cn). Yuang Chen and Hancheng Lu are also with Deep Space Exploration Laboratory, Hefei 230088, China. Chang Wen Chen is with the Department of Computing, The Hong Kong Polytechnic University, Hong Kong, China (e-mail: changwen.chen@polyu.edu.hk). 
}}
\maketitle

\begin{abstract}
In the Industrial Internet of Things (IIoT), the frequent transmission of large amounts of data over wireless networks should meet the stringent timeliness requirements. Particularly, the freshness of packet status updates has a significant impact on the system performance. In this paper, we propose an age-of-information (AoI)-aware multi-base station (BS) real-time monitoring framework to support extensive IIoT deployments. To meet the freshness requirements of IIoT, we formulate a joint task offloading and resource allocation optimization problem with the goal of minimizing long-term average AoI. Tackling the core challenges of combinatorial explosion in multi-BS decision spaces and the stochastic dynamics of IIoT systems is crucial, as these factors render traditional optimization methods intractable. Firstly, an innovative branching-based Dueling Double Deep Q-Network (Branching-D3QN) algorithm is proposed to effectively implement task offloading, which optimizes the convergence performance by reducing the action space complexity from exponential to linear levels. Then, an efficient optimization solution to resource allocation is proposed by proving the semi-definite property of the Hessian matrix of bandwidth and computation resources. Finally, we propose an iterative optimization algorithm for efficient joint task offloading and resource allocation to achieve optimal average AoI performance. Extensive simulations demonstrate that our proposed Branching-D3QN algorithm outperforms both state-of-the-art DRL methods and classical heuristics, achieving up to a 75\% enhanced convergence speed and at least a 22\% reduction in the long-term average AoI.
\end{abstract}

\begin{IEEEkeywords}
Age of Information (AoI); Delay-Sensitive Tasks; Deep Reinforcement Learning (DRL); Task Offloading; Resource Allocation; Industrial Internet of Things (IIoT).
\end{IEEEkeywords}

\section{Introduction}

\IEEEPARstart{W}{ith} the gradual maturation of fifth-generation mobile communication (5G) technologies and the large-scale commercial application of the Internet of Things (IoT) technology, modern communications have entered a new era of interconnected everything, and people's daily lives have undergone significant transformations \cite{11037391, 10255264, 9129776}. Currently, there are 18.8 billion active IoT devices globally, with Industrial IoT (IIoT) accounting for 60\% of new IoT device installations \cite{SQMagazine2025}. IIoT is expected to seamlessly integrate sensing, communication, and control in industrial automation by leveraging wireless communication networks \cite{vitturi2019industrial, bansal2020survey}.

\par In this context, due to the global trend of extensive IIoT deployments, future networks will inevitably face increased demands for computation and communication resources to effectively handle the massive data generated by these devices and the substantial computational and transmission burdens they impose \cite{bansal2020survey, 7900337}. Given the limited computation capacity and battery power of IoT devices, relying solely on local computation typically falls short of meeting these emerging task requirements \cite{dai2022task, huang2023aoi}. Generally, IIoT devices necessitate to wirelessly offload computational tasks to resource-rich locations for further processing. Traditional cloud computing paradigms address this by offloading computationally intensive and delay-sensitive tasks to the cloud \cite{chen2025dmsa}. However, due to the dispersed distribution of IIoT devices, the considerable distance between devices and the cloud center may result in severe path loss during wireless transmission, thereby leading to inefficient data transmission \cite{10460318,10529607}. Consequently, the centralized architecture of cloud computing is ill-suited for IIoT, as it leads to excessive cloud load, low resource utilization, and intolerable latency \cite{10876766}. These limitations highlight the necessity and potential of mobile edge computing (MEC), which brings computation closer to IIoT devices, thereby enabling low-latency and high-efficiency processing \cite{chen2025dmsa,10876766,10460318,10529607}.

\par Although offloading compute-intensive tasks from devices to nearby edge servers can effectively improve quality of service (QoS) performance, non-customized task offloading and resource allocation schemes cannot fully fulfill the QoS requirements of IIoT \cite{bansal2020survey, 7900337}, since the tasks involved in IIoT devices typically have strict timeliness requirements (i.e., age of information, AoI) \cite{10355071, hu2024timeliness}, which place higher demands on the computation and communication capabilities of networks. Specifically, the AoI at time $t$ is defined as the time elapsed since the generation of the most recently received information update, i.e., $\mathrm{AoI}(t) = t - G(t)$, where $G(t)$ denotes the generation time of the latest successfully received update \cite{yates2018age,yates2021age, abd2019role}. To this end, the system needs to accurately evaluate and optimize timeliness for the task offloading and resource allocation of real-time IIoT application scenarios, especially in terms of the timeliness of task and device information updates, to ensure that the system can make correct decisions based on the latest information. Therefore, how to effectively evaluate and optimize the timeliness of information in real-time systems has become a key issue in improving the performance of IoT systems \cite{lu2018age, guo2024aoi}.

\vspace{-1.2em}

\subsection{Motivations and Challenges}

\par The core of implementing an AoI-aware real-time monitoring system to facilitate IIoT access lies in accurately evaluating how task offloading and resource allocation strategies impact the system's AoI performance \cite{11089504,9900430,10589619,huang2023aoi,10301673,9530374}. Although extensive research on transmission optimization for delay-sensitive IIoT tasks in edge networks has established a foundation and provided effective solutions for managing IIoT device access, significant challenges persist in practical applications.

\par \textbf{Existing AoI models are oversimplified and fail to model and control dynamic system behaviors accurately.} Traditional AoI models typically assume fixed and static service parameters (e.g., service rate). However, in real-world IIoT networks, service capabilities fluctuate due to factors such as varying network conditions, resource contention, and device heterogeneity \cite{6195689,8406928}. By definition, the value of AoI at time $t$ depends not only on the current state but also on prior updates, reflecting the inherent temporal correlation between information freshness and system dynamics \cite{8406928, 9546792}. This stochastic and dynamic nature of AoI optimization makes it challenging to apply conventional convex optimization methods, necessitating the adoption of more flexible optimization tools and intelligent decision-making techniques.

\par \textbf{IIoT access in multi-BS-MEC scenarios exacerbates the complexity of task offloading and resource allocation.} Most existing studies predominantly focus on single BS-MEC environments \cite{10185607, 9530374, 7900337,10355071, hu2024timeliness}, failing to capture the prevalent multi-BS configurations, which are more practical in IIoT systems. Multi-BS-MEC systems provide enhanced computing and communication resources, which are essential for supporting the offloading requirements of IIoT devices \cite{10529607,10460318, 10360267}. However, under the dynamic and random AoI metric, determining the optimal offloading strategy for multi-BS-MEC systems becomes extremely complex. The challenge lies in efficiently managing the interplay between multi-BS-MEC and a dynamic IIoT environment, requiring algorithmic refinement for efficient offloading and transmission decisions that can ensure rapid and stable convergence.

\par \textbf{Existing learning algorithms exhibit significant limitations in addressing the curse of dimensionality in multi-BS-MEC systems with IIoT access.} Learning algorithms especially deep reinforcement learning (DRL) algorithms have been considered as an effective tool for solving MDPs \cite{mnih2015human, 9013982}. However, with these algorithms, task offloading strategies in IIoT typically involve discrete binary decision variables, which render applying traditional DRL methods challenging \cite{10091508}. Although Deep Q-Networks (DQN) were introduced to tackle discrete decision-making problems, they suffer from overestimation and maximization bias due to using the same neural networks for both action evaluation and selection \cite{mnih2015human}. Double DQN (DDQN) mitigates this issue by decoupling action selection from evaluation. However, it still struggles with distinguishing the true impact of actions when states exhibit highly similar or overlapping effects \cite{van2016deep}. Dueling DQN improves on this by decomposing the Q-function into state-value functions, further reducing bias \cite{sewak2019deep}. Dueling Double DQN (D3QN) algorithm combines the architectural benefits of Dueling DQN and the bias-reduction strategy of DDQN \cite{wang2016dueling, 10155465, zabihi2023reinforcement, 10660558}. When applied to multi-BS systems with high freshness required IIoT access, the exponential growth of the action space results in a severe curse of dimensionality, significantly hindering learning efficiency and convergence \cite{zabihi2023reinforcement}.

\vspace{-1em}

\subsection{Main Contributions}

\par In this paper, we propose an AoI-aware multi-BS-MEC real-time monitoring system to support IIoT access. The system aims to minimize long-term average AoI under the constraints of delay, bandwidth, computation, energy, and task scheduling. To efficiently solve the non-convex dynamic stochastic optimization problem, we decompose it into task offloading and resource allocation subproblems and develop corresponding optimization algorithms for their solutions. The primary contributions of this paper are summarized as follows:

\begin{itemize}
    \item We propose an AoI-aware multi-BS-MEC real-time monitoring system for IIoT access and formulate a dynamic stochastic optimization problem to minimize long-term average AoI under delay, bandwidth, computation capacity, energy consumption, and task scheduling constraints.

    \item To effectively tackle this problem, we decompose it into task offloading and resource allocation subproblems. For task offloading, we design an innovative branching structure into D3QN, named the Branching-D3QN (BD3QN) algorithm, which reduces the network complexity from exponential to linear levels and alleviates the curse of dimensionality caused by the increasing number of IIoT devices and BSs, enabling faster and more stable convergence.

    \item For resource allocation, we prove its convexity by deriving the semi-definite property of the Hessian matrix of bandwidth and computational resources with respect to the long-term average AoI. Then, we propose an iterative optimization algorithm for efficient joint task offloading and resource allocation to achieve the minimization of the system's long-term average AoI.

    \item Extensive simulations demonstrate that the proposed BD3QN algorithm significantly outperforms mainstream DRL algorithms (including DQN and D3QN), achieving up to 75\% faster convergence. Moreover, when integrated into the proposed joint iterative optimization scheme, it exhibits superior scalability and resource utilization efficiency, reducing the long-term average AoI by up to 22\% compared with DQN, D3QN, Greedy, and Random offloading schemes.
\end{itemize}

\par The remainder of this paper is organized as follows. Sec. II reviews related works. Sec. III introduces the proposed system model. Sec. IV formulates the dynamic stochastic optimization problem of minimizing the system's long-term average AoI. Sec. V provides the algorithm designs and solutions. Sec. VI provides extensive performance evaluations against the state-of-the-art benchmarks. Finally, Sec. VII concludes the paper and explores future directions.

\section{RELATED WORKS}

\par This section first reviews the literature on AoI-based timeliness research concerning task offloading and resource allocation in the context of IIoT. Then, we provide a comprehensive survey of DRL-driven task offloading and resource allocation schemes in MEC networks tailored to IIoT. Finally, we identify key limitations from these studies that are closely aligned with the aforementioned challenges.

\subsection{Research on AoI-based Timeliness for Task offloading and Resource Allocation}

\par As a critical metric for quantifying the data packet's freshness, AoI has been extensively adopted in time-sensitive IIoT systems \cite{yates2021age}. In \cite{10185607}, the authors analyzed the average AoI performance of MEC systems under local, edge, and partial computing modes, with implications for IIoT. The study in \cite{9665756} focused on minimizing the expected sum AoI in MEC-assisted IIoT networks by offloading resource-intensive tasks to edge servers, providing tractable analytical expressions for AoI. In \cite{10310169}, the authors proposed a three-layer multi-unmanned aerial vehicle (UAV) assisted MEC system, where joint optimization of task offloading, UAV trajectories, and communication resources was developed to minimize system AoI of IIoT under stochastic task arrivals. To enhance security and reliability in IIoT, the authors in \cite{11122470} incorporated blockchain into MEC and utilized AoI to quantify the freshness of status updates across multiple stakeholders. The authors in \cite{11089504} used the AoI metric to evaluate data freshness of the end-edge-cloud computing systems, where a joint optimization of task offloading and resource allocation was proposed to minimize both AoI and energy consumption under constraints of deadlines and capacity constraints. Moreover, to meet real-time demands of computation-intensive tasks, the study in \cite{9900430} leveraged AoI as a key criterion for edge-offloading decisions in IIoT systems, proposing a joint detection and offloading strategy to maximize computation throughput under AoI constraints.

\subsection{Studies on DRL-driven task offloading and resource allocation schemes in MEC Networks}

\par In recent years, artificial intelligence (AI) technologies, particularly deep reinforcement learning (DRL), have been extensively applied in MEC networks, significantly promoting the intelligence and adaptability of IIoT systems \cite{10589619,7900337, bansal2020survey, wu2025physiological, 11197211}. In \cite{huang2023aoi}, the authors investigated an AoI-aware joint energy control and computation offloading problem and developed a DRL-based algorithm to adapt to dynamic IIoT environments. To address task dependencies in IIoT, the authors in \cite{10301673} introduced a directed acyclic graph (DAG)-based computation offloading framework and formulated a joint optimization of delay, energy consumption, and AoI, solved using an enhanced D3QN approach. The authors in \cite{10510274} studied an aerial-ground collaborative MEC system to balance AoI and energy consumption, where a multi-objective proximal policy optimization algorithm was proposed to jointly optimize UAV trajectories and task offloading ratios, supporting IIoT aerial monitoring. Similarly, the authors in \cite{10181012} developed a decentralized multi-agent DRL framework for UAV crowdsensing to minimize both AoI and violation ratios while maximizing collected data under energy constraints. Furthermore, in \cite{10820863}, the authors leveraged a soft actor–critic DRL method to jointly optimize UAV trajectory and network configuration for AoI minimization under realistic cost and operational constraints in IIoT contexts.

\subsection{Limitations of these Existing Research Works}

\par Although considerable research has advanced AoI-aware MEC systems and DRL-based optimization methods, significant limitations remain in scaling these approaches to IIoT networks within multi-BS-MEC real-time monitoring environments. Most existing studies rely on oversimplified queueing models with static assumptions, which fail to capture the stochastic and dynamic characteristics of real-time IIoT traffic. Moreover, current works predominantly address single-BS scenarios, overlooking the complex coupling effects and combinatorial explosion of task offloading and resource allocation in multi-BS systems. Furthermore, state-of-the-art DRL algorithms exhibit poor scalability in high-dimensional action spaces, resulting in exponential computational overhead and unstable convergence. These limitations underscore the urgent need for scalable, AoI-aware optimization frameworks and robust DRL architectures capable of accurately modeling timeliness, enabling intelligent task offloading and resource allocation, and achieving efficient coordination across multi-BS MEC networks under dynamic IIoT workloads.

\section{System Model}

\par AS shown in Fig. \ref{fig_system_model}, we proposed an AoI-aware multi-BS-MEC real-time monitoring system to facilitate IIoT access. Each BS is equipped with an MEC server, referred to as BS-MEC, jointly providing communication and computation resources. The system comprises $N \geq 1$ single-antenna IIoT devices capable of real-time environmental sending and update-packet generation. Unlike conventional delay-based metrics that only capture the perspective of transmitting devices, they fail to reflect the timeliness experienced at the receiving end \cite{yates2018age}. Therefore, they are inadequate for evaluating real-time monitoring systems. To address this, we adopt the concept of AoI as the performance metric, enabling an accurate assessment of information freshness across the multi-BS-MEC architecture \cite{abd2019role,yates2021age}.

\begin{figure}[h]
\centering
\includegraphics[scale=0.45]{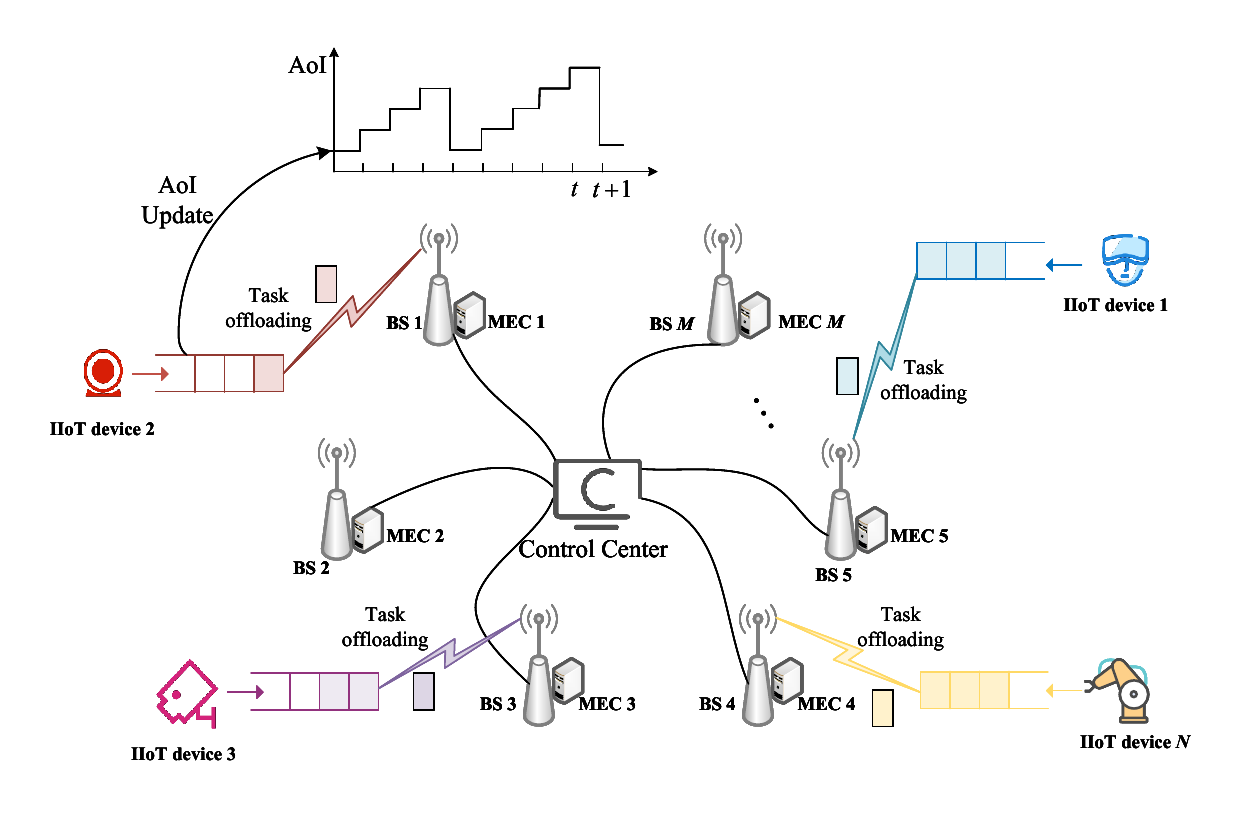}
\caption{\small The proposed AoI-aware multi-BS real-time monitoring system for IIoT access.}
\label{fig_system_model}
\end{figure}

\par Let $\mathcal{M}=\left\{1,2, \cdots, M \right\}$ and $\mathcal{N}=\left \{1,2, \cdots, N\right\}$ denote the index sets of BE-MEC servers and IIoT devices, respectively, with each device $i \in \mathcal{N}$ assigned a priority $\alpha_{i}$. Task arrivals at each IIoT device follow a Bernoulli distribution, and the system operates over $T$ time slots of duration $\tau_{max}$, indexed by $\mathcal{T}=\left \{1,2, \cdots, T \right\}$. We consider that each device adopts a passive sampling strategy, generating status update packets only upon environmental changes. Denote $z_{i}$ and $c_{i}$ as the task size and required CPU cycles for device $i$, respectively. To preserve real-time responsiveness, each device maintains a buffer queue where newly generated tasks are placed at the front. Next, we analyze the task offloading process in the proposed system model. To facilitate tractable analysis and focus on the core problems in multi-BS-MEC systems, orthogonal frequency bands are allocated among different BSs to eliminate inter-cell interference. To further mitigate inter-user interference, we adopt orthogonal frequency division multiple access, allowing each IIoT device to occupy a distinct orthogonal subchannel during transmission. The offloading matrix $\boldsymbol{\mathrm{X}}$ can be given by Eq. (\ref{eq1}), as follows:

\begin{equation}\label{eq1}
    \boldsymbol{\mathrm{X}} = \left[\boldsymbol{a}_{1}, \boldsymbol{a}_{2}, \cdots, \boldsymbol{a}_{N}\right] =
    \begin{bmatrix}
       \tilde{\boldsymbol{a}}_{1}\\
       \tilde{\boldsymbol{a}}_{1}\\
       \cdots\\
       \tilde{\boldsymbol{a}}_{N}
    \end{bmatrix}.
\end{equation}

\par Each IIoT device can offload its task to at most one BS-MEC in each time slot $t$. Let $\tilde{a}_{i}$ denote the offloading vector of device $i$, and $a_{i,j}^{t}$ be the binary decision variable indicating whether device $i \in \mathcal{N}$ offloads to BS $j \in \mathcal{M}$. The offloading succeeds with probability $\varepsilon$ and fails with probability $1-\varepsilon$. The wireless channel comprises small-scale Rayleigh fading $g_{i,j}(t) \sim \mathcal{CN}(0,1)$, so that the power gain $|g_{i,j}(t)|^2$ is exponentially distributed with unit mean, and large-scale path loss $d_{i,j}^{-\tau}$, where $d_{i,j}$ is the distance between device $i$ and BS $j$, and $\tau$ is the path-loss exponent. Assuming additive white Gaussian noise with power $\sigma^2$, the achievable transmission rate is given by Shannon's Capacity, as follows:

\begin{equation}\label{eq2}
    R_{i,j} (t) = B_{i,j}(t) \log_{2}\left(1 + \frac{p_{i}(t) |h_{i,j}(t)|^{2}}{\sigma^{2}}\right),
\end{equation}
where $B_{i,j}(t)$ and $p_{i}(t)$ represent the bandwidth and transmission power allocated to device $i$, respectively, and $h_{i,j}(t) = g_{i,j}(t) d_{i,j}^{-\tau/2}$ indicates the channel coefficient.

\par During task offloading, the total service time at each BS-MEC comprises two parts: (i) the wireless transmission delay over the air interface, and (ii) the computation delay incurred by task processing at the BS-MEC. Consequently, the transmission delay of device $i$ is given by

\begin{equation}\label{eq3}
    t_{i,j}^{trans} = \frac{z_{i}}{R_{i,j} (t)}.
\end{equation}

\par Then, the energy consumption resulted by device $i$ during wireless transmission is denoted as follows:

\begin{equation}\label{eq4}
    E_{i}(t) = \sum\limits_{j = 1}^{M} p_{i}(t) t_{i,j}^{trans}.
\end{equation}

\par Upon receiving the task from device $i$, the computation delay at BS-MEC $j$ is given by

\begin{equation}\label{eq5}
    t_{i,j}^{comp} = \frac{z_{i} c_{i}}{f_{i,j}(t)},
\end{equation}
where $f_{i,j}(t)$ denotes the computation resources allocated by BS-MEC $j$ to device $i$. Accordingly, the total processing delay of device $i$ is expressed as

\begin{equation}\label{eq6}
     D_{i}(t) = \sum_{j=1}^{M} a_{i,j}^{t}\left(t_{i,j}^{trans} + t_{i,j}^{comp}\right).
\end{equation}

\par To evaluate the real-time performance, this paper leverages the concept of AoI, defined as the time elapsed since the last successfully processed status update at the BS-MEC. Although prior studies have explored AoI via queuing models such as M/M/1 and D/M/1 \cite{6195689,8406928}, these works rely on restrictive assumptions about arrival and service distributions. In contrast, task arrivals and service processes in real-time monitoring systems are inherently random and asynchronous, making queuing-based formulations intractable. Unlike these traditional approaches constrained by predefined models, this paper directly employs AoI as a performance metric, providing a more flexible and realistic framework for analyzing and optimizing real-time task transmission, as illustrated in Fig.~\ref{fig_AoI_changes}.

\begin{figure}[htbp]
\centering
\includegraphics[scale=1]{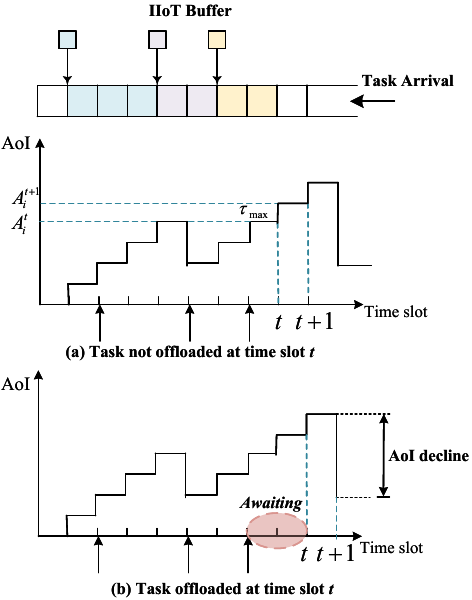} \caption{\small The AoI changes of the proposed multi-BS-MEC real-time monitoring system for IIoT access.}
\label{fig_AoI_changes}
\end{figure}

\par \textbf{\textit{Case 1}} (a), if the task of device $i$ has been successfully offloaded but not yet processed, or if the device remains unoffloaded or experiences transmission failure, the AoI at the next time slot can be denoted as follows:

\begin{equation}\label{eq7}
A_{i}^{t+1} = A_{i}^{t} + \tau_{max},
\end{equation}
where the AoI at slot $t + 1$ increases by one slot length $\tau_{max}$.

\par \textbf{\textit{Case 2}} (b), if the task has been offloaded and successfully completed, the AoI decreases as a fresh update becomes available only after both transmission and computation are finished. In this case, the AoI at slot $t+1$ can be expressed as

\begin{equation}\label{eq8}
     A_{i}^{t + 1} = t - G_{i}^{t} + t_{i,j}^{trans} + t_{i,j}^{comp},
\end{equation}
where $t$ denotes the current time, $G_{i}^{t}$ indicates the task generation time, and $t_{i,j}^{trans}$ and $t_{i,j}^{comp}$ represent the transmission and computation delays, respectively. Accordingly, the AoI evolution for device $i$ under both cases is summarized in (\ref{eq9}), and its dynamic behavior is illustrated in Fig.~\ref{fig_AoI_changes}.

\begin{equation}\label{eq9}
A_{i}^{t+1} = \left\{
\begin{array}{ll}
A_{i}^{t} + 1, & \hbox{Case 1;} \\
t - G_{i}^{t} + t_{i,j}^{trans} + t_{i,j}^{comp}, & \hbox{Case 2.}
\end{array}
\right.
\end{equation}

\section{Problem Formulation}

\par To optimize the timeliness of the proposed AoI-aware multi-BS-MEC real-time monitoring system, we formulate a long-term average AoI minimization problem, subject to strict constraints on delay, bandwidth, computation capacity, energy consumption, and task scheduling. The problem is formulated as follows:

\vspace{-2em}

\begin{subequations}\label{eq10}
\begin{align}
\mathcal{P}1: & \min_{\{B_{i,j}(t), f_{i,j}(t), a_{i,j}^{t}\}_{\forall i \in \mathcal{N}, j \in \mathcal{M}}} \lim_{T \rightarrow \infty} \frac{1}{T} \sum\limits_{i = 1}^{N} \sum\limits_{j = 1}^{M} \sum_{t=1}^{T} \alpha_{i} A_{i}^{t}, \\
\emph{s.t.} \ \ \ & \sum\limits_{j=1}^{M} a_{i,j}^{t} \leq 1, a_{i,j}^{t} \in \{0,1\}, \quad \forall i \in \mathcal{N}, \\
& \sum\limits_{i=1}^{N} a_{i,j}^{t} \leq K, \quad \forall j \in \mathcal{M}, \\ & E_{i}(t) \leq E_{max}, \quad \forall i \in \mathcal{N},\\
& \sum\limits_{i=1}^{N} B_{i,j}(t) \leq B_{j}^{max}, \quad \forall j \in \mathcal{M}, \\
& \sum\limits_{i=1}^{N} f_{i,j}(t) \leq f_{j}^{max}, \quad \forall j \in \mathcal{M}, \\
& D_{i}(t) \leq \tau_{max}, \quad \forall i \in \mathcal{N},
\end{align}
\end{subequations}
constraints (\ref{eq10}b) and (\ref{eq10}c) ensure that each device offloads its tasks to at most one BS-MEC per time slot, while limiting the number of simultaneously served devices. Constraint (\ref{eq10}d) enforces the per-device energy budget, crucial for battery-powered IIoT nodes. Constraints (\ref{eq10}e) and (\ref{eq10}f) bound the total bandwidth and computational resources of each BS-MEC, and (\ref{eq10}g) guarantees timely completion of task transmission and execution within one slot. Notably, problem $\mathcal{P}1$ involves discrete variables $a_{i,j}^{t}$ and continuous variables $B_{i,j}(t)$ and $f_{i,j}(t)$. Their complex coupling renders $\mathcal{P}1$ a non-convex dynamic stochastic optimization problem, for which traditional convex optimization methods are inapplicable.

\section{Algorithm Designs and Solutions}

\par The primary challenges in solving this long-term dynamic stochastic optimization problem lie in the dependence of AoI on both current state and prior states. Moreover, the optimization objective of $\mathcal{P}1$ involves discrete offloading decisions with continuous variables, which significantly exacerbates the solution's computational complexity. To simplify the solving process, we decouple the original problem into two subproblems: task offloading and resource allocation.

\vspace{-0.5em}

\begin{figure}[htbp]
\centering
\includegraphics[scale=0.4]{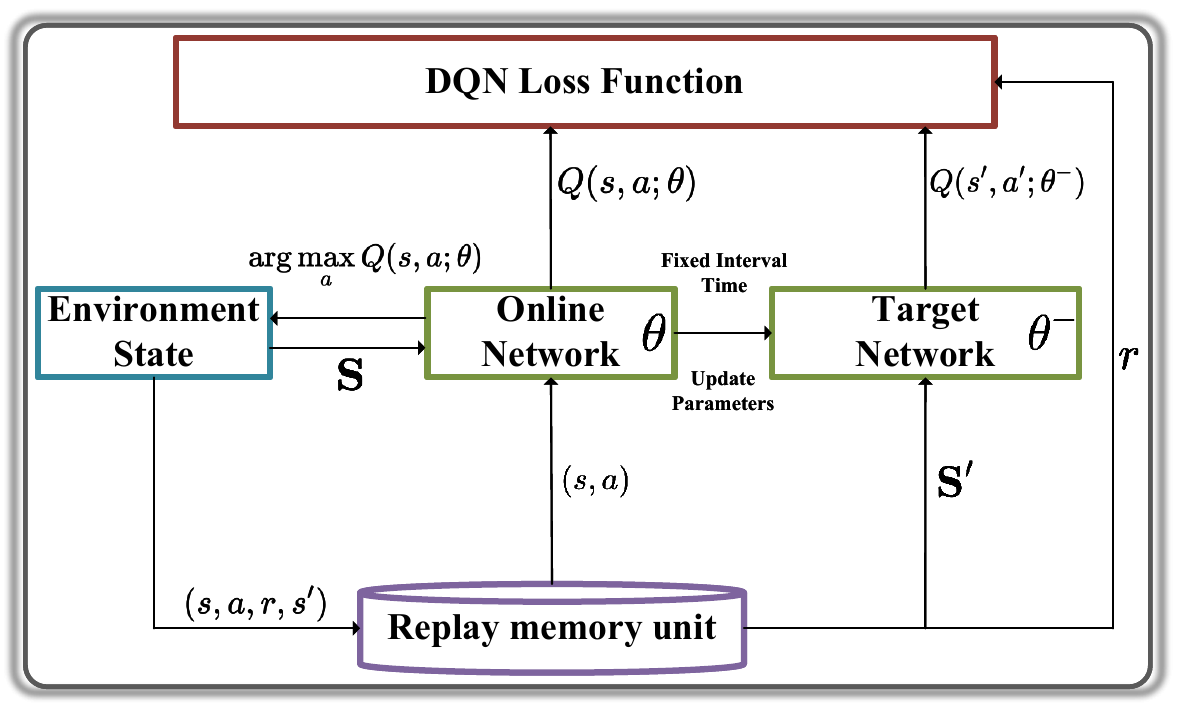}
\caption{The basic framework of DQN framework \cite{mnih2015human}.}
\vspace{-2em}
\label{fig_DQN}
\end{figure}

\subsection{The Solution for the Task Offloading Problem}

\par To effectively tackle the task offloading subproblem, we first reformulate it as an equivalent Constrained Markov Decision Process (CMDP), wherein the current state is influenced by past states and also determines future states. This CMDP is encapsulated by a multidimensional array $\left[\mathbb{A}, \mathbb{S}, \boldsymbol{\mathbf{R}}, \boldsymbol{\mathbf{P}}, \boldsymbol{\mathbf{\pi}}\right]$, with elements defined as follows:

\begin{enumerate}
    \item \underline{\textbf{$\mathbb{A}$: Action Space}} represents interactions between BS-MEC and IIoT devices via offloading variables $a_{i,j}^{t}$. The action space $\mathbb{A}$ is expressed as the offloading matrix $\boldsymbol{\mathrm{X}}$.

    \item \underline{\textbf{$\mathbb{S}$: State Space}} comprises four components:
        \begin{itemize}
            \item The information of AoI for all IIoT devices at time slot $t$ can be denoted as $\boldsymbol{A}(t) = \left[A_{1}^{t}, A_{2}^{t}, \cdots, A_{N}^{t}\right]$.

            \item The information of energy consumption for all devices at time slot $t$ is represented by $\boldsymbol{E}(t) = \left[E_{1}(t), E_{2}(t), \cdots, E_{N}(t)\right]$.

            \item The information of delay for all IIoT devices at time slot $t$ can be expressed as $\boldsymbol{D}(t) = \left[D_{1}(t), D_{2}(t), \cdots, D_{N}(t)\right]$.

            \item The information of channel state information (CSI) for the considered system model can be denoted as $\boldsymbol{H}(t) = \left\{h_{i,j}(t)\right\}_{i \in \mathcal{N}, j \in \mathcal{M}}, t \in \mathcal{T}$.
        \end{itemize}

    \item \underline{\textbf{$\boldsymbol{\mathbf{R}}$: Reward function}} hinges on both the actions and states at time slot $t$. To align the reward with the optimization objective of $\mathcal{P}1$, we define the reward function as the negative weighted sum of AoI:

    \begin{equation}\label{eq11}
         r(t) = - \sum\limits_{i = 1}^{N}\alpha_{i} A_{i}^{t}.
    \end{equation}

    \item \underline{\textbf{$\boldsymbol{\mathbf{P}}$: State transition probability}} quantifies the transition probability of advancing to state $s(t+1)$ from $s(t)$ upon action execution.

    \item \underline{\textbf{$\boldsymbol{\mathbf{\pi}}$: Strategy}} denotes the state evolution over  $\mathcal{T}$, denoted as $\theta = \left\{\pi_{1}, \cdots, \pi_{T}\right\}$, with the optimal strategy:

    \begin{equation}\label{eq12}
        \pi^{\ast} = \arg \min\limits_{\boldsymbol{\mathbf{\pi}}}\frac{1}{T} \mathbb{E}_{\boldsymbol{\mathbf{\pi}}}\left[\sum_{t=1}^{T}c(t) | s(1)\right].
    \end{equation}
\end{enumerate}

\begin{figure*}[t]
\centering
\includegraphics[scale=0.7]{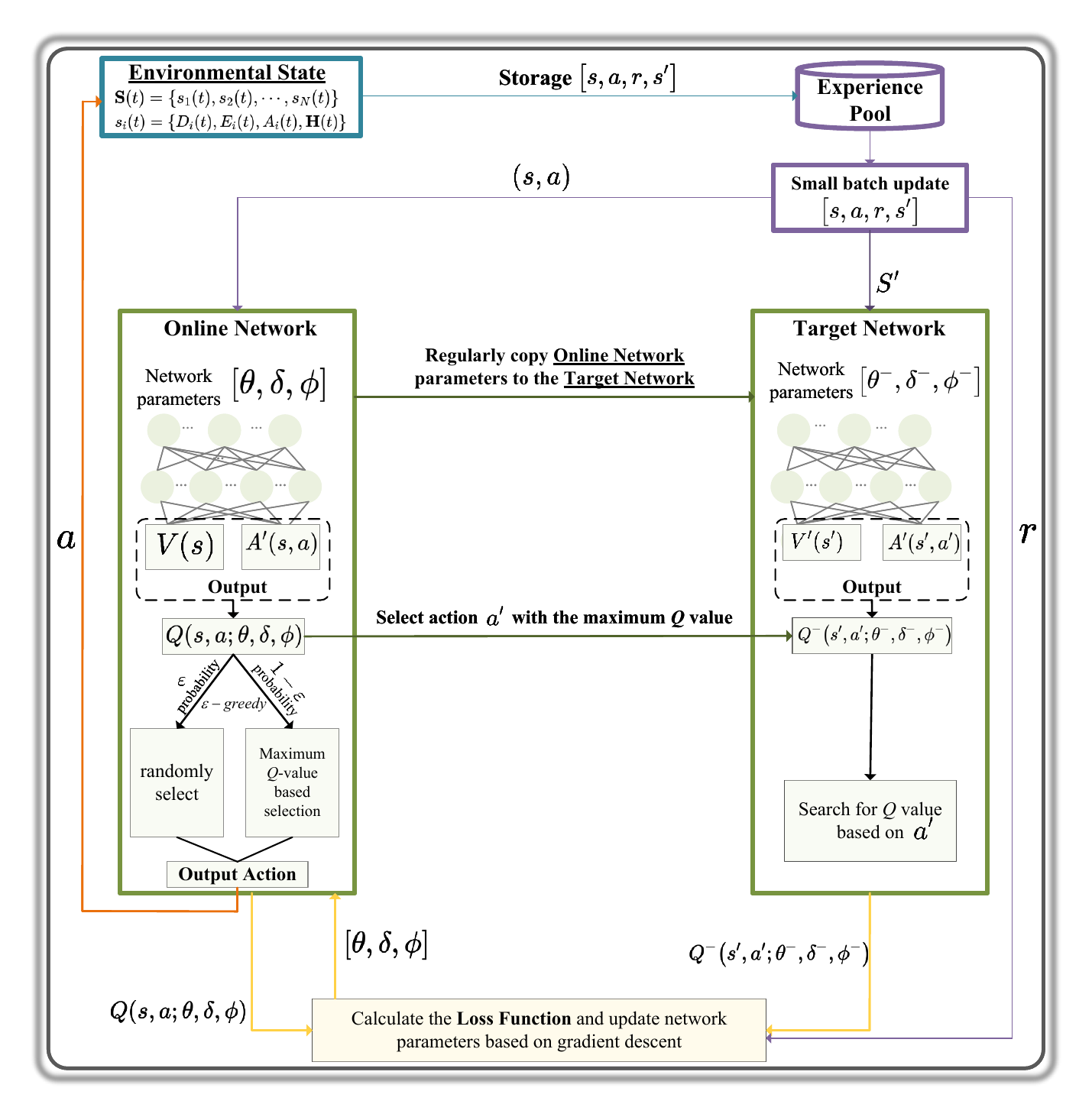}
\caption{The D3QN framework designed for our considered task offloading subproblem.}
\label{fig_D3QN}
\vspace{-2em}
\end{figure*}

\par Owing to the CMDP's inherent intractability, we transform CMDP into an unconstrained MDP by introducing relaxation terms into the reward function (\ref{eq11}), which seamlessly integrates objectives with implicit constraint enforcement, and the refined reward function becomes:

\begin{equation}\label{eq13}
    r(t) = - \sum\limits_{i=1}^{N} \bigg[\alpha_{i} A_{i}^{t} - \zeta \left(E_{max} - E_{i}(t)\right) - \beta \left(\tau_{max} - D_{i}(t)\right)\bigg].
\end{equation}

\par Solving CMDP is equivalent to solving the unconstrained MDP. This equivalence facilitates DRL deployment to optimize long-term utility of the system, quantified via Q-values, also known as the action value function, which estimates the expected returns of taking a specific action and following a certain strategy under a given state. The Bellman equation can be used to characterize the recursive relationship between the action value at time slot $t$ and the expected value of future actions, thus enabling recursive calculation of the Q-value. The specific expression of the Bellman equation can be denoted as

\begin{equation}\label{eq14}
    Q\left(s_{t}, a_{t}\right) = r(t) + \gamma \mathbb{E} \left[Q\left(s_{t+1}, a_{t+1}\right) \right],
\end{equation}
with discount factor $\gamma \in \left[0,1\right]$ balancing immediate versus prospective rewards. Smaller values of $\gamma$ emphasize short-term gains, whereas $\gamma \rightarrow 1$ prioritizes long-term returns. For the discrete binary action space $\mathbb{A}$, Deep Q-Network (DQN) integrates reinforcement learning with deep neural networks to approximate Q-tables, adeptly navigating high-dimensional states \cite{mnih2015human}. In this setting, each device acts as an agent, and the framework of DQN is depicted in Fig.~\ref{fig_DQN}.

\par As illustrated in Fig.~\ref{fig_DQN}, the DQN framework consists of two key components: the online network with parameter $\theta$ and the target network with parameter $\theta^{-}$. The online network estimates Q-values based on the current state $\mathbf{S}$, with its parameters updated iteratively from samples $(s,a)$ in the pool of replay memory unit. The target network, updated periodically by synchronizing with the online network, stabilizes the training process by mitigating target value fluctuations. This decouple network structure significantly improves the instability caused by target movement during the learning process, enhancing the learning and training stability \cite{mnih2015human}. The DQN loss function is formulated as the mean squared error (MSE) between the predicted Q-value $Q(s,a;\theta)$ and the target Q-value $Q(s^{\prime},a^{\prime};\theta^{-})$ \cite{10380323}:

\begin{equation}\label{eq15}
    L\left(\theta\right) = \mathbb{E}\bigg[\bigg(r(t) + \gamma \max_{a^{\prime}} Q(s^{\prime}, a^{\prime}; \theta) - Q(s, a; \theta) \bigg)^{2} \bigg].
\end{equation}

\par Although DQN achieves superior reinforcement learning, it incurs overestimation and maximization bias since the same network is used for action selection and evaluation, leading to overestimation of the maximum Q value in certain states. Double DQN (DDQN) decouples these two processes: (i) the online network selects the action with the highest Q-value, while (ii) the target network evaluates it, leveraging its stability to reduce bias \cite{van2016deep}. Then the loss function of the DDQN network is given by

\begin{equation}\label{eq16}
    \begin{aligned}
       L\left(\theta\right) = \mathbb{E}\bigg[\bigg(r(t) + \gamma Q^{-1}\left(s^{\prime}, \arg \max_{a^{\prime}} Q(s^{\prime}, a^{\prime}; \theta); \theta\right) & \\
        - Q(s, a; \theta) \bigg)^{2} \bigg]& .
    \end{aligned}
\end{equation}

\par DDQN alleviates overestimation issues, it struggles with action discrimination in noisy or reward-similar scenarios, and its reliance on a single-state representation also fails to capture factors influencing decisions. In this case, Dueling DQN innovates by decomposing the Q-function into a state value function $V(s)$ and an advantage function $A^{\prime}(s, a)$, rewritten as $Q\left(s,a\right) = V(s) + A^{\prime}(s, a)$, with two parallel output structures parameterized by $(\phi, \delta)$ \cite{sewak2019deep, wang2016dueling}. However, directly summing $V$ and $A^{\prime}$ obscures their individual contributions, reducing network performance. To address this, a modified Q-function with a penalty term is introduced \cite{wang2016dueling}. Denoting the set of devices offloading tasks at time slot $t$ as $B(t)$, the modified Q-function is reformulated as follows:

\begin{equation}\label{eq17}
     \begin{aligned}
         Q\left(s,a; \theta, \delta, \phi\right) & = V(s; \delta, \lambda) + \bigg[A^{\prime} (s,a; \theta, \phi) - \\
         & \frac{1}{\left|B(t)\right|} \sum\limits_{a^{\ast} \in B(t)} A^{\prime} \left(s, a^{\ast}; \theta, \phi\right)\bigg],
     \end{aligned}
\end{equation}
where $V$ and $A^{\prime}$ are relatively unchanged with different inputs when Q is determined, improving the recognizability issues.

\par Synergizing the refined estimation of Dueling DQN with the bias reduction of DDQN, the Dueling Double DQN (D3QN) algorithm is proposed \cite{10155465, zabihi2023reinforcement, 10660558}. D3QN alleviates overestimation while more accurately estimating the state-action values, making it well-suited for the task offloading problem considered here and theoretically closer to the optimal solution. Its loss function follows the DDQN form, while Q-value calculations adopt the Dueling architecture, as shown in (\ref{eq18}), where target network parameters $\left(\theta^-, \delta^-, \phi^-\right)$ update periodically by copying the corresponding online network parameters $\left(\theta, \delta, \phi\right)$. However, D3QN faces scalability hurdles in expansive action spaces, which grow exponentially with IIoT devices and BSs, significantly slowing learning and convergence. The D3QN architecture tailored for the task offloading problem is illustrated in Fig. \ref{fig_D3QN}.

\vspace{-0.5em}

\begin{figure}[h]
\centering
\includegraphics[scale=0.65]{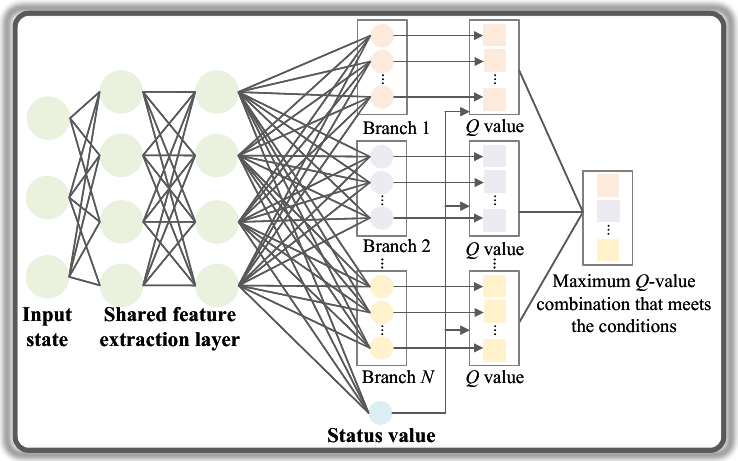}
\caption{The framework of the Branching-D3QN algorithm.}
\label{fig_branching_D3QN}
\vspace{-1em}
\end{figure}

\begin{figure*}[t]
\centering
\hrulefill
\begin{equation}\label{eq18}
  \begin{aligned}
    L\left(\theta, \delta, \phi\right) = \mathbb{E} \bigg[\bigg(r_{t} + \gamma Q^{-} \big(s^{\prime}, \arg \max_{a^{\prime}} Q^{-}(s^{\prime}, a^{\prime}; \theta, \delta, \phi); \theta^{-1}, \delta^{-1}, \phi^{-1} \big) - Q(s, a; \theta, \delta, \phi)\bigg)^{2} \bigg].
  \end{aligned}
\end{equation}
\hrulefill
\begin{equation}\label{eq19}
     \begin{aligned}
          L\left(\theta, \delta, \phi\right) & = \mathbb{E} \bigg[\frac{1}{N} \sum\limits_{n=1}^{N}  \bigg(r_{t} + \gamma \max_{n} Q^{-}_{n} \big(s_{t+1}, \arg \max_{a^{\prime}_{n} \in \mathcal{A}_{n}} Q_{n}(s_{t+1}, a^{\prime}_{n}; \theta, \delta, \phi); \theta^{-1}, \delta^{-1}, \phi^{-1} \big) - Q_{n}(s, a_{n}; \theta, \delta, \phi)\bigg)^{2} \bigg].
     \end{aligned}
\end{equation}
\hrulefill
\end{figure*}

\par To address the issue of exponential growth in the action space dimensionality, this paper proposes an innovative enhancement for the D3QN framework by incorporating a branching structure, yielding the novel BD3QN algorithm. As illustrated in Fig.~\ref{fig_branching_D3QN}, BD3QN can significantly mitigate the complexity of the decision-making process by decomposing the high-dimensional action space into smaller and more manageable subspaces. The core innovation of BD3QN lies in its ability to conceptualize different features (or dimensions) of each action as independent branches, wherein each branch optimizes decisions specific to its corresponding feature. This design not only reduces the number of output variables per neural network but also enhances the learning efficiency by parallelizing the optimization across multiple subspaces. Therefore, the proposed BD3QN algorithm effectively addresses the curse of dimensionality, fostering scalable and faster convergence for the proposed AoI-aware multi-BS-MEC real-time monitoring systems with massive IIoT access.  In particular, when deploying the BD3QN algorithm for task offloading, the offloading decision of each device is regarded as an independent branch, yielding $N$ branches in aggregate. Each branch has $M+1$ feasible actions, corresponding to the selection of a specific BS-MEC for offloading or local execution. The branching structure empowers our proposed BD3QN algorithm to concurrently optimize the offloading decisions of each device in parallel, effectively addressing the issues of expanding action space caused by escalating numbers of IIoT devices and BSs. As a result, the time complexity of the solution becomes $N(M+1)+1$. Moreover, the detailed branch structure can be illustrated in Fig. \ref{fig_branching_D3QN}. After integrating this branch structure, the Q-value for each sub-branch (e.g., branch $n$) can be denoted as

\vspace{-2em}

\begin{equation}\label{eq20}
   \begin{aligned}
       Q\left(s,a_{n}; \theta, \delta, \phi\right) & = V(s; \theta, \delta) + \bigg[A^{\prime} (s,a_{n}; \theta, \phi) - \\
         & \frac{1}{M+1} \sum\limits_{a^{\ast}_{n} \in \mathcal{A}_{n}} A^{\prime} \left(s, a^{\ast}_{n}; \theta, \phi\right)\bigg].
   \end{aligned}
\end{equation}

\par In this case, the loss value of BD3QN can be reformulated as (\ref{eq19}), and the specific algorithm process can be detailed in \textbf{Algorithm 1}.

\vspace{0.5em}

\begin{algorithm}[htbp]
\small
 \setstretch{0.9}
        \caption{The Branching-D3QN (BD3QN)-based Task Offloading Algorithm.}
        \KwIn{IIoT device set $\mathcal{N}$, BS set $\mathcal{M}$, channel status of devices $h_{n}$, delay threshold $\tau_{max}$, energy constraint of device $E_{max}$, AoI weights of devices $\alpha_{i}$, parameter update cycle $\boldsymbol{\Gamma}$;}
        \textbf{Initialize:} Environment status $\boldsymbol{S}(t)$, parameters of online network $\left[\theta, \delta, \phi \right]$, and target network $\left[\theta^{-}, \delta^{-}, \phi^{-} \right]$, $\varepsilon$, and $\varepsilon_{min}$;\\
        \For{$t = 0, 1, \cdots, T$}{
           Generate a random number $\kappa \in (0,1)$;\\
           \eIf{$\kappa \geq \varepsilon$}
             {
               Select the action combination with the maximum Q value;\\
               \tcp{\!\!\!\!\textbf{Calculate Q value and select the optimal action combination.}}
             }
             {
               Select actions randomly;
             }
             \tcp{\!\!\!\!\textbf{Update the AoI of each IIoT device.}}
               Update AoI for devices with action changes according to (\ref{eq9})-Case 2;\\
               Update AoI for other devices without action changes based on (\ref{eq9})-Case 1;\\
             \If{\textnormal{Devices with unchanged actions AoI} $\geq$ \textnormal{AoI}$_{max, i}$}
             {
               Update the AoI of this device to \textnormal{AoI}$_{max, i}$;\\
               Increase the selection probability of this device;\\
             }
               Execute action $a(t)$, observe reward function $r(t)$ and next state $s(t+1)$;\\
               Store the combination $\left[s(t), r(t), s(t+1)\right]$ into the experience pool;\\
               Randomly sample a batch from the experience pool for training;\\
               Calculate the Loss function based on Eq. (\ref{eq19});\\
               \tcp{\!\!\!\!\textbf{Update the online network parameters.}}
               Update online network parameters $\left[\theta, \sigma, \phi\right]$ via using gradient descent method;\\
               \tcp{\!\!\!\!\textbf{Update the target network parameters.}}
              \If{\textnormal{Training Steps} $\% \Gamma == 0$}
             {
               Copy the online network parameters to the target network;\\
               Update target network parameters: $\theta^{-} = \theta$, $\sigma^{-} = \sigma$, and $\phi^{-} = \phi$;\\
             }
          }
       \KwOut{Optimal task offloading strategy $\boldsymbol{\pi}^{\ast}$.}
\end{algorithm}

\vspace{-0.5em}

\subsection{The Solution for Resource Allocation Problem}

\par Based on the task offloading strategies $a_{i,j}^{t}$ obtained by \textbf{Algorithm 1}, the original problem $\mathcal{P}1$ can be transformed into $\mathcal{P2}$, as follows:

\vspace{-1em}

\begin{subequations}\label{eq21}
    \begin{align}
              \mathcal{P}2: & \min_{\{B_{i,j}(t), f_{i,j}(t)\}_{\forall i \in \mathcal{N}, j \in \mathcal{M}}} \lim_{T \rightarrow \infty} \frac{1}{T} \sum\limits_{i = 1}^{N} \sum\limits_{j = 1}^{M} \sum_{t=1}^{T} \alpha_{i} A_{i}^{t}, \\
                \emph{s.t.} \ \ \ & E_{i}(t) \leq E_{max}, \quad \forall i \in \mathcal{N},\\
                                  & \sum\limits_{i=1}^{N} B_{i,j}(t) \leq B_{j}^{max}, \quad \forall j \in \mathcal{M}, \\
                                  & \sum\limits_{i=1}^{N} f_{i,j}(t) \leq f_{j}^{max}, \quad \forall j \in \mathcal{M}, \\
                                  & D_{i}(t) \leq \tau_{max}, \quad \forall i \in \mathcal{N}.
    \end{align}
\end{subequations}

\par Decomposing $\mathcal{P}2$ into a resource allocation decision-making process with granularity for each given time slot $t \in \mathcal{T}$:

\vspace{-1em}

\begin{subequations}\label{eq22}
    \begin{align}
              \mathcal{P}3: & \min_{\{B_{i,j}(t), f_{i,j}(t)\}_{\forall i \in \mathcal{N}, j \in \mathcal{M}}} \sum\limits_{i = 1}^{N} \sum\limits_{j = 1}^{M} \alpha_{i} A_{i}^{t}, \\
                \emph{s.t.} \ \ \ & (\ref{eq21}b), (\ref{eq21}c), (\ref{eq21}d), \rm{and}\  (\ref{eq21}e).
    \end{align}
\end{subequations}

\par Expanding the objective function (\ref{eq22}a), and $\mathcal{P}3$ is equivalently transformed into $\mathcal{P}4$, as follows:

\begin{subequations}\label{eq23}
    \begin{align}
              \mathcal{P}4: & \min_{\{B_{i,j}(t), f_{i,j}(t)\}_{\forall i \in \mathcal{N}, j \in \mathcal{M}}} \sum\limits_{i = 1}^{N} \alpha_{i} \bigg\{ \sum\limits_{j = 1}^{M} a_{i,j}^{t} \bigg(t - G_{i}^{t} + \nonumber \\
                            & \quad \quad \quad \frac{d_{i}}{B_{i,j}(t) \log_{2} \big(1 + \frac{p_{i}(t)|h_{i,j}(t)|^{2}}{\sigma^{2}}\big)} + \frac{d_{i} c_{i}}{f_{i,j}(t)} \bigg)\bigg\}, \\
                \emph{s.t.} \ \ \ & (\ref{eq21}b), (\ref{eq21}c), (\ref{eq21}d), \rm{and}\  (\ref{eq21}e) \nonumber.
    \end{align}
\end{subequations}

\par Next, we demonstrate through Theorem \ref{theo1} that $\mathcal{P}4$ is a strictly convex optimization problem, as follows:

\begin{myTheo}\label{theo1}
    The objective function of $\mathcal{P}4$ is convex regarding the resource allocation variables ${B_{i,j}(t), f_{i,j}(t)}_{\forall i \in \mathcal{N}, j \in \mathcal{M}}$, and its all constraints are convex. Thus, $\mathcal{P}4$ is a strictly convex optimization problem.
\end{myTheo}

\begin{proof}
  The objective function (\ref{eq23}) of $\mathcal{P}4$ is denoted as $F$. To demonstrate its convexity, we derive the first and second derivatives of $F$ with respect to the variables $B_{i,j}(t)$ and $f_{i,j}(t)$. Firstly, for $f_{i,j}(t)$:

  \begin{equation}\label{eq24}
    \nabla F_{\{f_{i,j}(t), \forall i \in \{i | a_{i,j}^{t} = 1\}\}} = -\frac{\alpha_{i} d_{i} c_{i}}{(f_{i,j}(t))^2} < 0,
  \end{equation}

  \begin{equation}\label{eq25}
     \nabla^2 F_{\{f_{i,j}(t), \forall i \in \{i | a_{i,j}^{t} = 1\}\}} = \frac{2 \alpha_{i} d_{i} c_{i}}{(f_{i,j}(t))^3} > 0,
   \end{equation}

   Next, for $B_{i,j}(t)$:

   \begin{equation}\label{eq26}
     \nabla F_{\{B_{i,j}(t), \forall i \in \{i | a_{i,j}^{t} = 1\}\}} \!=\! \frac{- \alpha_{i} d_{i}}{(B_{i,j}(t))^2 \log_2 \!\left( 1 + \frac{p_i(t) h_{i,j}(t)}{\sigma^2} \right)} \!<\! 0,
    \end{equation}

    \begin{equation}\label{eq27}
      \nabla^2 F_{\{B_{i,j}(t), \forall i \in \{i | a_{i,j}^{t} = 1\}\}} \!=\! \frac{2 \alpha_{i} d_{i}}{(B_{i,j}(t))^3 \log_2 \!\left( 1 + \frac{p_i(t) h_{i,j}(t)}{\sigma^2} \right)} \!>\! 0.
    \end{equation}

    Since $F$ is separable in terms of $B_{i,j}(t)$ and $f_{i,j}(t)$, the mixed partial derivatives are zero. Hence, the Hessian matrix of $F$ is given by:

    \begin{equation}\label{eq28}
        \boldsymbol{\mathrm{H}}(\boldsymbol{f}, \boldsymbol{B}) =
          \begin{bmatrix}
            \frac{2 \alpha_{i} d_{i} c_{i}}{(f_{i,j}(t))^3} & 0 \\
            0 & \frac{2 \alpha_{i} d_{i}}{(B_{i,j}(t))^3 \log_2 \left( 1 + \frac{p_i(t) h_{i,j}(t)}{\sigma^2} \right)}
          \end{bmatrix}
    \end{equation}

    Given that the diagonal elements of the Hessian matrix are all positive for $\boldsymbol{f} > 0$ and $\boldsymbol{B} > 0$, this implies that the matrix is strictly diagonally dominant. Therefore, the Hessian matrix is semi-positive definite, which confirms that the objective function $F$ is convex over the domain $\boldsymbol{f} > 0$ and $\boldsymbol{B} > 0$. Consequently, $\mathcal{P}4$ is a strictly convex optimization problem for the variables $\boldsymbol{f}$ and $\boldsymbol{B}$.
\end{proof}

\par According to \textbf{Theorem \ref{theo1}}, the problem $\mathcal{P}3$ can be easily tackled using the CVX tool.

\subsection{Joint Task Offloading and Resource Allocation Scheme}

\par In Sections V-A and V-B, the original problem was decoupled into two subproblems: task offloading and resource allocation, each equipped with a tailored optimization strategy. To efficiently address their coupling, we adopt an alternating optimization (AO) framework, where the proposed BD3QN-based offloading algorithm and the resource allocation scheme are executed iteratively until convergence. The detailed procedure is summarized in \textbf{Algorithm~2}.

\begin{algorithm}[htbp]
\small
\setstretch{1.0}
\caption{AO-Based Joint Task Offloading and Resource Allocation Scheme}
\KwIn{
    IIoT device set $\mathcal{N}$, BS set $\mathcal{M}$, device channel states $h_n$, delay threshold $\tau_{\text{max}}$, energy limit $E_{\text{max}}$, AoI weights $\alpha_i$, parameter update cycle $\boldsymbol{\Gamma}$.
}
\KwOut{
    Optimized strategies $\left\{B_{i,j}^{\ast}(t), f_{i,j}^{\ast}(t), a_{i,j}^{\ast}(t)\right\}$ for all $i \in \mathcal{N}, j \in \mathcal{M}, t \in \mathcal{T}$.
}
\textbf{Initialize:}
\begin{itemize}
    \item Environment state $\boldsymbol{S}(t)$;
    \item Online network parameters $\left[\theta, \delta, \phi\right]$ and target network parameters $\left[\theta^{-}, \delta^{-}, \phi^{-}\right]$;
    \item Initial feasible resource allocation and offloading scheme $\{B_{i,j}(t), f_{i,j}(t), a_{i,j}(t)\}$ satisfying constraints $(\ref{eq10}b)$-$(\ref{eq10}g)$.
\end{itemize}
\Repeat{convergence or maximum iterations reached}{
    Execute \textbf{Algorithm 1} to update task offloading strategy $\{a_{i,j}(t)\}$;\\

    Solve subproblem $\mathcal{P}4$ via convex optimization (e.g., CVX) to update $\{B_{i,j}(t), f_{i,j}(t)\}$;\\

    Update the joint strategy $\{B_{i,j}(t), f_{i,j}(t), a_{i,j}(t)\}$ for the next iteration;\\
}
\end{algorithm}

\section{Performance Evaluation}

\par In this section, we conduct extensive simulation experiments to rigorously evaluate the effectiveness of the proposed schemes. Specifically, we benchmark the proposed BD3QN algorithm against both classical baselines and state-of-the-art algorithms, and provide detailed comparative analyses. Furthermore, comprehensive ablation studies are performed to verify the contribution of each key component and to demonstrate the overall superiority of our proposed schemes.

\subsection{Experimental Parameter Settings}

\par In our simulations, the IIoT devices and BSs are randomly distributed within a $1000~m \times 1000~m$ rectangular region. All experiments are executed on servers equipped with Nvidia RTX 3090Ti GPU to guarantee efficient training and evaluation. The simulation experiment is implemented in Python 3.9, while the training process leverages TensorFlow 2.0. The detailed simulation parameters used throughout our performance evaluation are summarized in Table \ref{table_I}.

\begin{table}[h]
 \renewcommand{\arraystretch}{1.0}
 \caption{\small Simulation Parameter Settings}
 \vspace{-0.5em}
 \label{table_I}
 \centering
 \resizebox{0.95\columnwidth}{!}
 {
 \begin{tabular}{|l|c|}
  \hline
  \bfseries Parameter                                & \bfseries Value\\
  \hline
  System bandwidth $B$                               & $400$ kHz\\
  \hline
  Limitation of each BS's device access number $K$   & $3$\\
  \hline
  Large-scale fading model $d_{i,j}^{-\tau}$         & $128.1 + 37.6 \log(d_{i,j}[\rm{km}])$\\
  \hline
  Noise power spectral density $N_{0}$               & $-174 \ \rm{dBm/Hz}$\\
  \hline
  Task package size $z_{i}$                          & $[1,3]$ Mbits\\
  \hline
  Maximum delay constraint $\tau_{max}$              & $1$ s\\
  \hline
  MEC server computing capability $f_{j}^{max}$      & $[7,10]$ Gbps\\
  \hline
  Maximum power limit of IIoT devices $p_{i}^{\max}$ & $27.8 \ \rm{dBm/Hz}$\\
  \hline
  Maximum energy limit of IIoT devices $E_{max}$     & $[0.5,1.5]$ J\\
  \hline
  Update steps $\Gamma$                              & $100$\\
  \hline
  Pool size                                          & $50000$\\
  \hline
  Batch size $|\mathcal{B}|$                         & $64$\\
  \hline
  Learning rate                                      & $64$\\
  \hline
  Decay rate of learning rate                        & $0.95$\\
  \hline
  Decay steps of learning rate                       & $10000$\\
  \hline
  Activation function                                & $\rm ReLU$\\
  \hline
  Optimizer                                          & $\rm Adam$ \\
  \hline
  Episode                                            & $1000000$ \\
  \hline
  Discount-related coefficients $\varsigma, \beta$   & [0.8, 0.95]\\
  \hline
 \end{tabular}
 }
\end{table}

\begin{figure*}[t]
    \centering
    \begin{minipage}[b]{0.47\textwidth}
        \centering
        \begin{subfigure}{\linewidth}
            \centering
            \includegraphics[width=\linewidth]{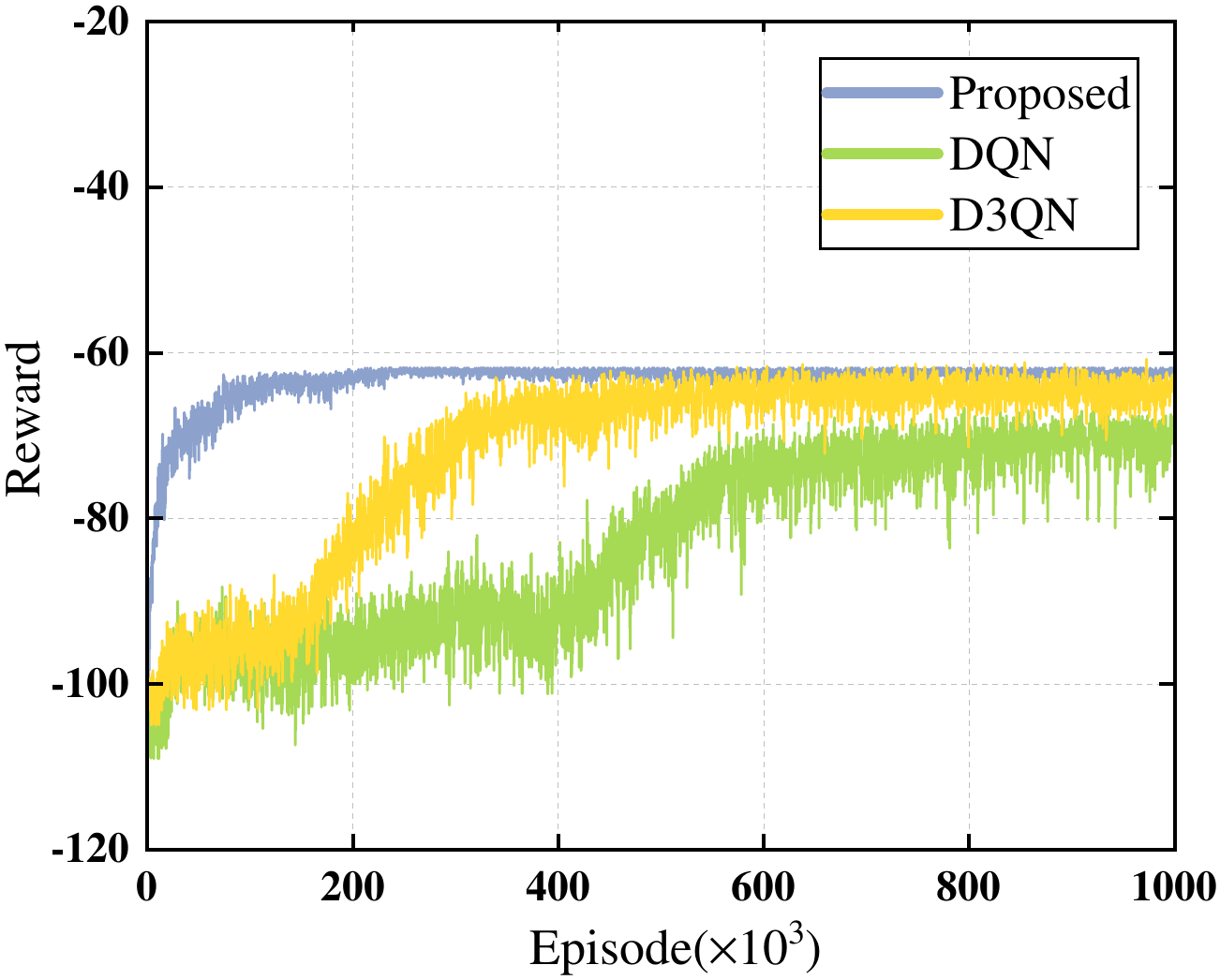}
            \caption{Reward changes over episodes.}
        \end{subfigure}
    \end{minipage}
    \hspace{0.005\textwidth}
    \begin{minipage}[b]{0.47\textwidth}
        \centering
        \begin{subfigure}{\linewidth}
            \centering
            \includegraphics[width=\linewidth]{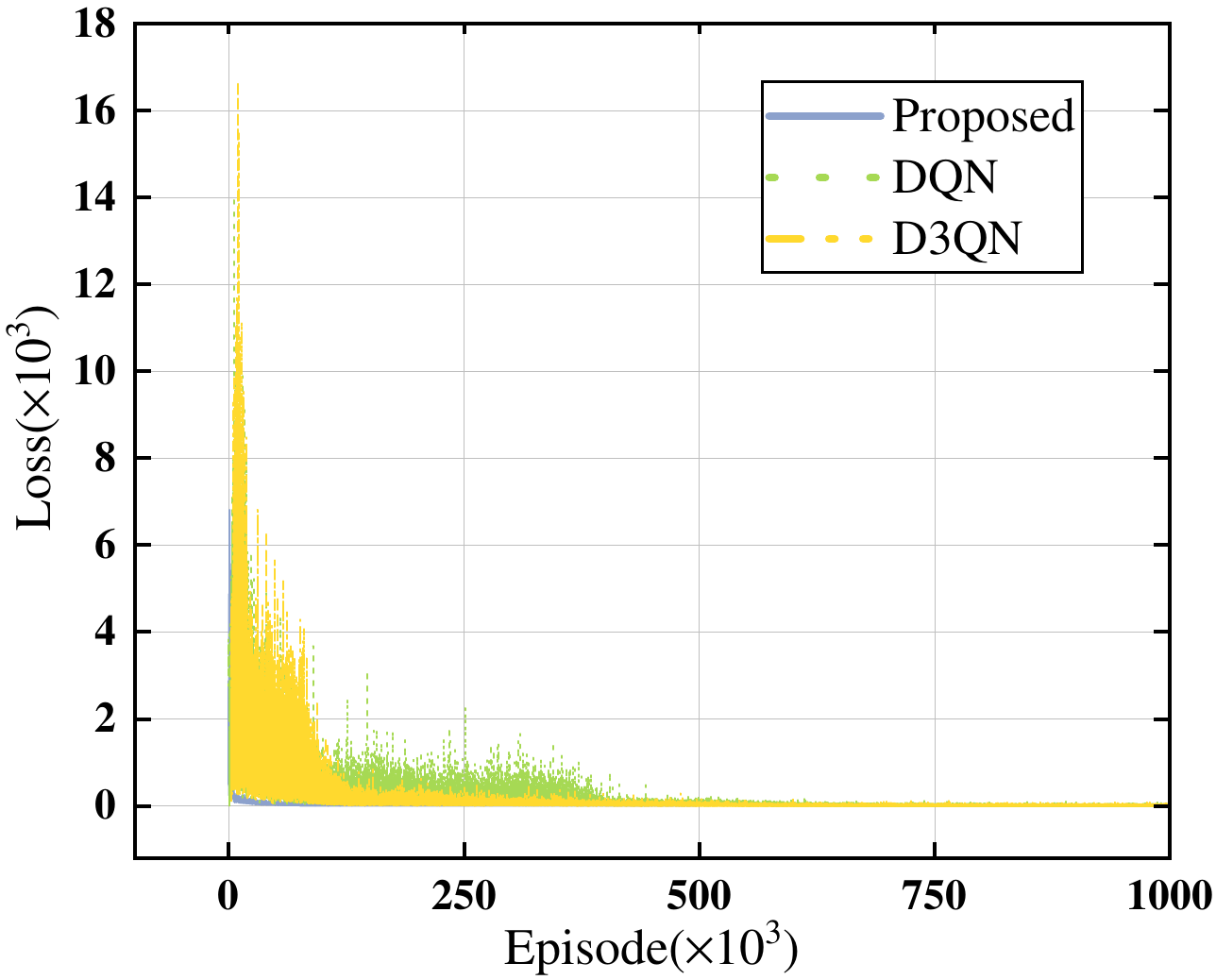}
            \vspace{-1.5em}
            \caption{Loss changes over episodes.}
        \end{subfigure}
    \end{minipage}
    \caption{\small Performance Comparison of the proposed BD3QN Algorithm with DQN and D3QN. (a) Reward changes over episodes. (b) Loss changes over episodes.}
    \label{fig_convergence}
\end{figure*}

\begin{figure*}[t]
    \centering
    \begin{minipage}[b]{0.47\textwidth}
        \centering
        \begin{subfigure}{\linewidth}
            \centering
            \includegraphics[width=\linewidth]{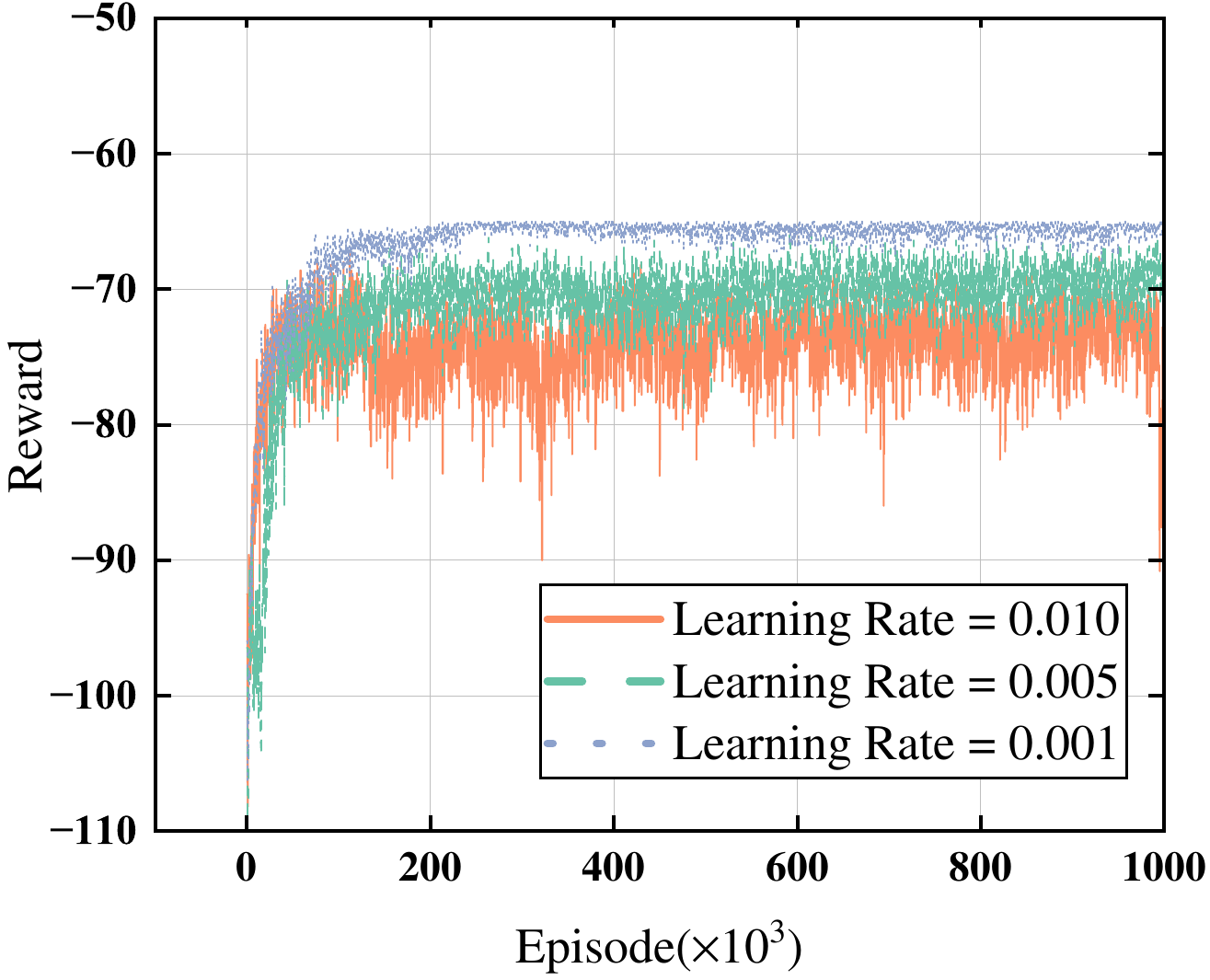}
            \caption{Reward changes over episodes.}
        \end{subfigure}
    \end{minipage}
    \hspace{0.005\textwidth}
    \begin{minipage}[b]{0.47\textwidth}
        \centering
        \begin{subfigure}{\linewidth}
            \centering
            \includegraphics[width=\linewidth]{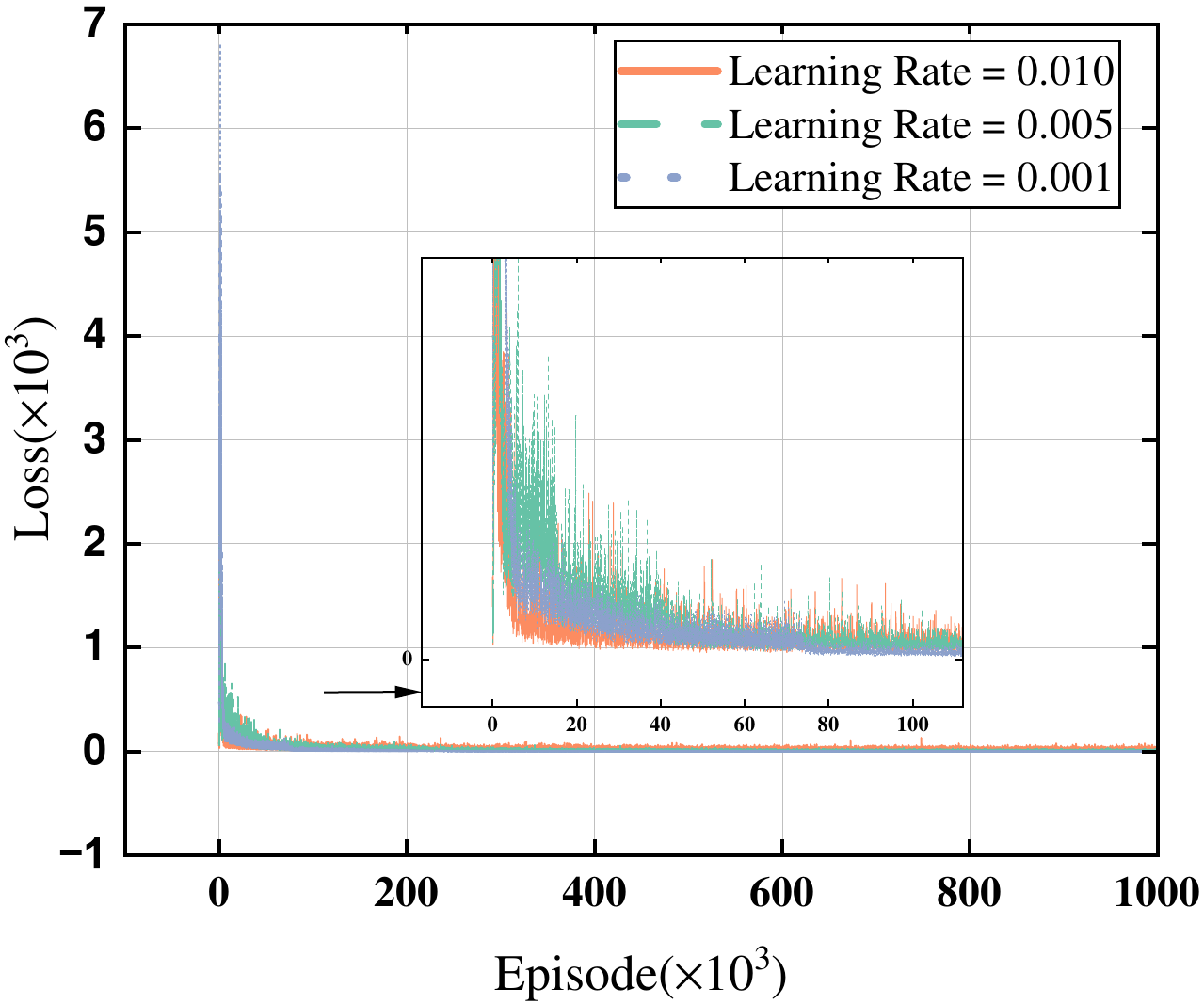}
            \vspace{-1.5em}
            \caption{Loss changes over episodes.}
        \end{subfigure}
    \end{minipage}
    \caption{\small The convergence performance verification of the proposed BD3QN algorithm at different learning rates. (a) Reward changes over episodes. (b) Loss changes over episodes.}
    \vspace{-2em}
    \label{fig_convergence_vs_learning_rate}
\end{figure*}

\subsection{Baseline Schemes for Comparison}

\par To comprehensively evaluate the effectiveness of the proposed BD3QN algorithm and joint task offloading and resource allocation scheme (denoted as \textit{\textbf{Proposed}}), we compare it with several representative baseline schemes, including classical heuristics and state-of-the-art DRL-based algorithms:

\begin{itemize}
    \item \textbf{Greedy Algorithm} \cite{8968760}: A heuristic method that prioritizes IIoT devices with higher AoI values for task offloading, while allocating idle BSs accordingly to minimize the overall AoI.

    \item \textbf{Random Offloading}: A benchmark strategy in which idle BSs randomly select IIoT devices for task offloading, serving as a performance lower bound.

    \item \textbf{DQN Algorithm} \cite{10380323, mnih2015human}: A DRL-based method that employs the standard DQN architecture to optimize the long-term average AoI without considering overestimation or state–action value decomposition.

    \textbf{D3QN Algorithm} \cite{10155465, zabihi2023reinforcement, 10660558}: An enhanced DRL approach that integrates DDQN and dueling DQN network architectures to mitigate overestimation and improve state–action value estimation for long-term average AoI minimization.
\end{itemize}

\vspace{-1.0em}

\subsection{Experimental Results}

\vspace{-0.5em}

\subsubsection{Convergence Verification}

\par We first evaluate the convergence performance of the proposed BD3QN algorithm and compare it with the popular DQN and D3QN algorithms. Convergence is a prerequisite for ensuring the stability and reliability of reinforcement learning models. In this experiment, the average reward is adopted as the evaluation metric. In reinforcement learning, the definition of an episode is usually defined as a complete sequence of interactions between the agent and the environment until reaching a terminal state. This process can be expressed as a complete sequence of <state--action--reward>, denoted as $<s_{1}, a_{1}, r_{1}, s_{2}, a_{2}, r_{2}, \cdots, s_{n}, a_{n}, r_{n}>$. Each episode thus represents a complete instance of the task from the start to the end, during which the agent continuously learns and adapts to the environment.

\par As illustrated in Fig. \ref{fig_convergence} (a) and (b), the proposed BD3QN algorithm exhibits a significantly faster and more stable convergence behavior compared with both DQN and D3QN algorithms. It can be observed that the proposed BD3QN algorithm achieves at least a $75\%$ improvement in convergence speed with respect to both reward and loss, clearly highlighting its efficiency advantage. This remarkable improvement stems from our proposed branching structure, which reduces the action space complexity from the exponential level $(M+1)^{N}$ to the linear level $(M+1)N+1$, thereby significantly alleviating the curse of dimensionality and accelerating convergence. Meanwhile, D3QN shows better convergence than DQN due to its dueling architecture, which effectively distinguishes between state values and action advantages, leading to more accurate Q-value estimation. Nevertheless, its performance still lags far behind the proposed BD3QN algorithm. These results collectively confirm that the proposed BD3QN algorithm not only achieves faster convergence but also maintains superior stability throughout training.

\par Furthermore, Fig. \ref{fig_convergence_vs_learning_rate} illustrates the impact of learning rate on the convergence behavior of the proposed BD3QN algorithm. These experimental results demonstrate that the learning rate critically influence both convergence speed and stability. A relatively small learning rate allows for fine-grained parameter updates, leading to stable and precise convergence, but requires more episodes to reach steady performance. Conversely, an excessively high learning rate accelerates early convergence but often induces oscillations near the optimum or even divergence, thereby degrading overall performance. This observation underscores the inherent trade-off between convergence speed and stability in learning rate selection. In extreme cases, an overly small learning rate may also trap the algorithm in local optima, preventing convergence within a limited number of training episodes.

\subsubsection{Performance Comparison}

\par Next, we comprehensively compare the proposed BD3QN algorithm with baseline schemes in terms of average AoI performance. In DRL methods, it is common to observe the impact of different batch sizes on the performance and training efficiency of the model by changing the batch size during the training process \footnote{Batchsize denotes the number of training samples per iteration. Smaller batch sizes enhance randomness and help escape local optima but increase gradient variance and instability, while larger batch sizes yield more stable gradients and faster convergence at the cost of higher resource demands and potential premature convergence.}.

\par As shown in Fig. \ref{fig_AAoI_vs_Batchsize}, we analyze the impact of batch size on the average AoI under varying numbers of IIoT devices. Firstly, it can be observed that under the same Batchsize, the more IIoT devices accessed, the worse the average AoI performance of the system. Moreover, the results show that increasing batch size initially improves performance by reducing gradient variance and making updates closer to the optimal direction. However, once the gradient estimation becomes sufficiently accurate, further increases yield diminishing returns and may even restrict parameter exploration, reducing adaptability. In practice, batch sizes of 64 or 128 strike a good balance between training efficiency and model performance, making appropriate batch size selection crucial for effective model training.

\begin{figure}[htbp]
\centering
\includegraphics[scale=0.5]{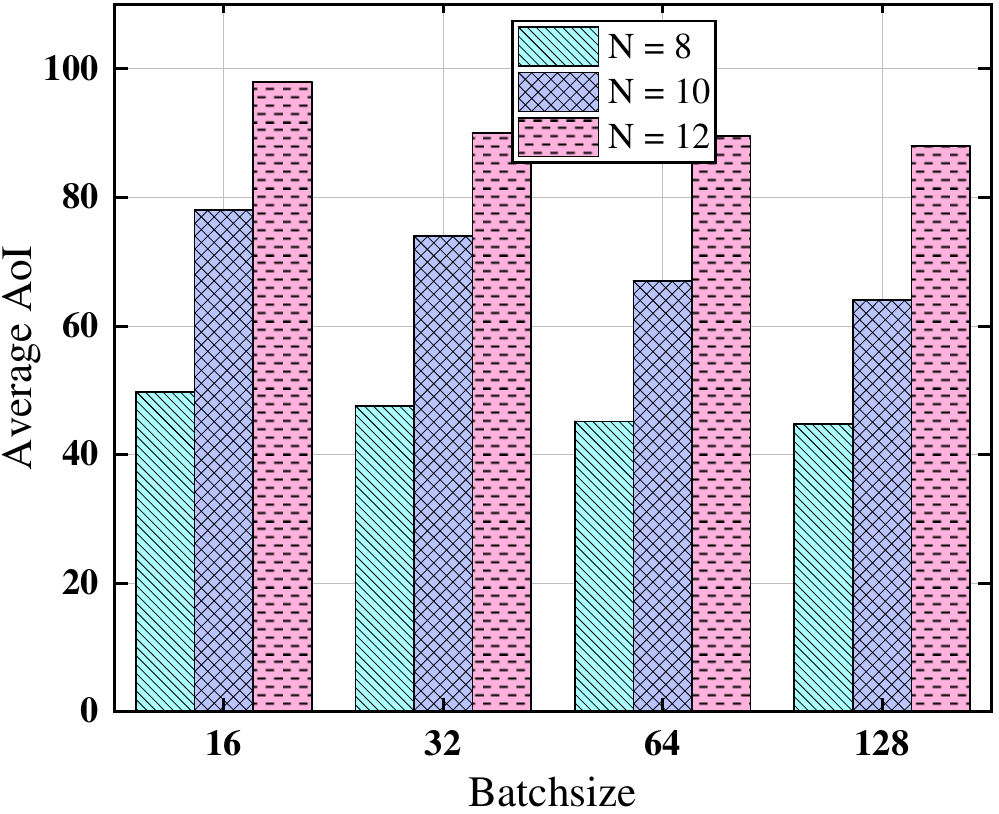}
\caption{The variation of average AoI with Batchsize size under different numbers of IIoT devices.}
\label{fig_AAoI_vs_Batchsize}
\end{figure}

\par In real-time monitoring systems, the task arrival rate $\lambda$ ($0 \leq \lambda \leq 1$) critically influences AoI performance. The task arrival rate quantifies the frequency of new task generation per unit time, where higher values accelerate task generation, thereby increasing device update frequency. In this case, the system can process and respond to the latest data faster, further optimizing system performance and user experience. Conversely, a lower $\lambda$ reduces update frequency, leading to relatively outdated information on the device, increased average AoI, and diminished the timeliness and efficiency of the system. As shown in Fig.~\ref{fig_AAoI_vs_Task_arrival_rate}, we conduct experiments to demonstrate that average AoI decreases with the increase of $\lambda$ since frequent data generation maintains data freshness at device nodes. Furthermore, it can be observed that the proposed BD3QN algorithm significantly outperforms benchmarks, achieving the 2\% average AoI reduction over D3QN, 10.4\% over DQN, and 17.4\% over the Greedy algorithm at $\lambda = 0.8$, underscoring the superior performance gains of the proposed algorithm compared to other baseline schemes in terms of high task arrival rates, data freshness, and system efficiency.

\begin{figure}[htbp]
\centering
\includegraphics[scale=0.5]{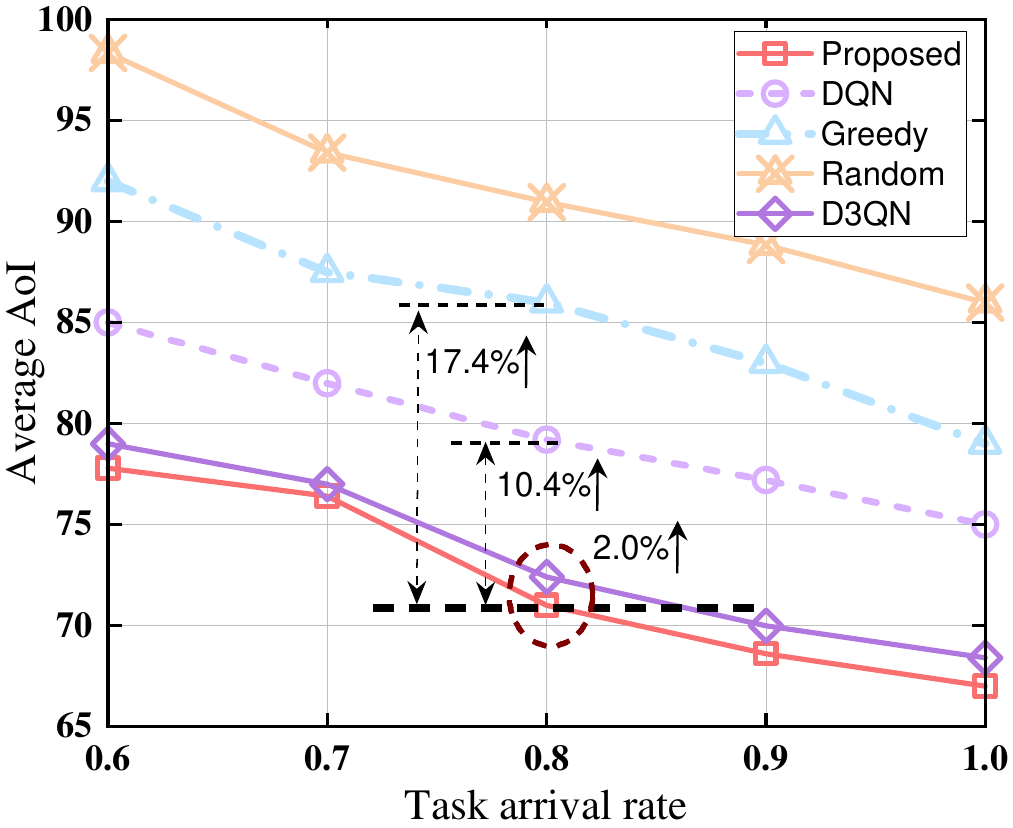}
\caption{Performance comparison of average AoI variation with different task arrival rates.} 
\label{fig_AAoI_vs_Task_arrival_rate}
\end{figure}

\begin{figure}[htbp]
\centering
\includegraphics[scale=0.4]{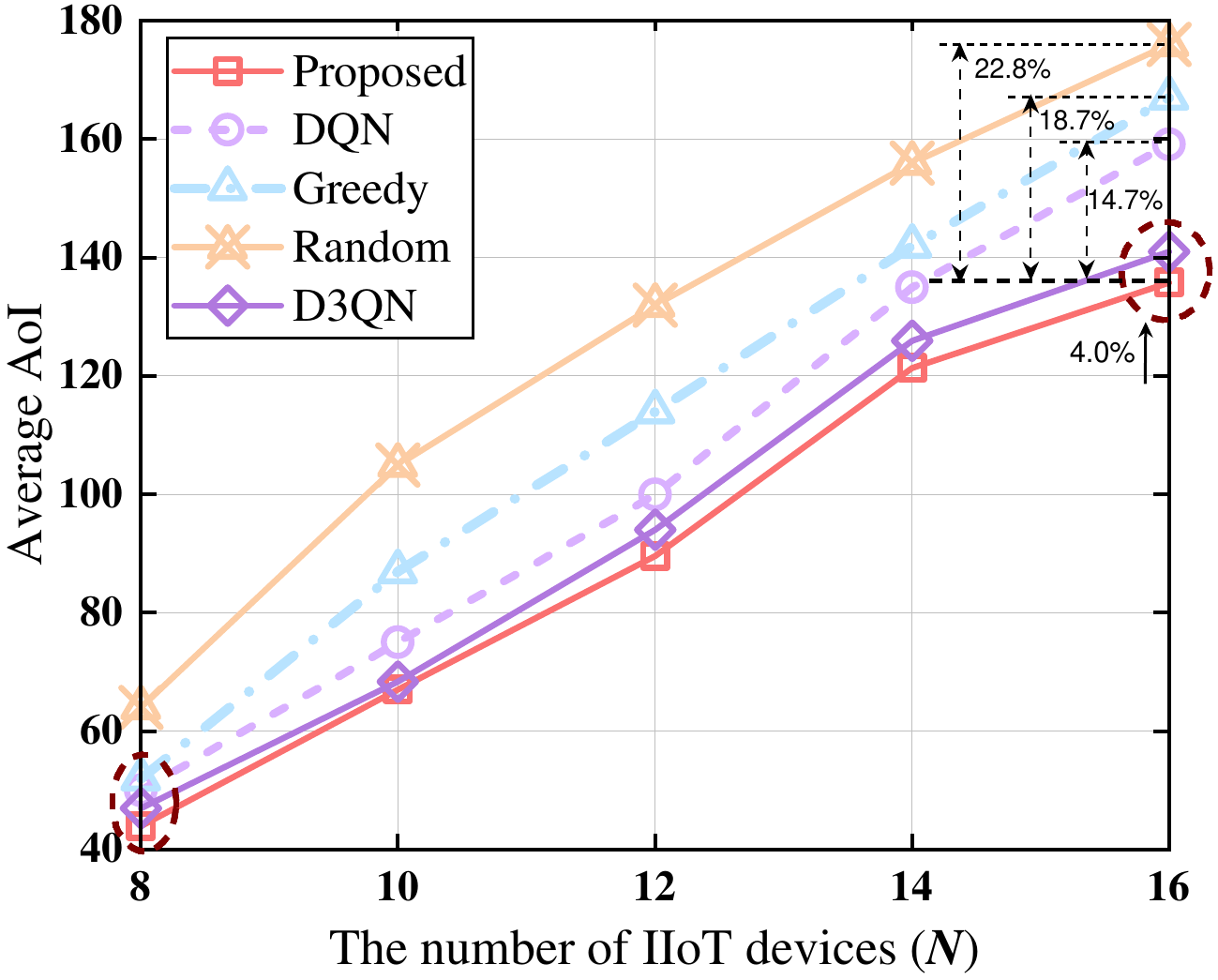}
\caption{\small Performance comparison of average AoI with different numbers of IIoT devices.}
\label{fig_AAoI_vs_IIoT}
\end{figure}

\par As illustrated in Fig. \ref{fig_AAoI_vs_IIoT}, we investigate the average AoI performance of the system under different numbers of IIoT devices. It can be observed that as the number of IIoT devices increases, the system's average AoI also grows accordingly due to heightened resource contention. However, compared to other baseline schemes, the proposed BD3QN algorithm demonstrates superior performance. Specifically, when the scale of IIoT devices is small (e.g., $N = 8$), compared to D3QN, DQN, and Greedy, the performance gains achieved by the proposed schemes in average AoI are not as pronounced; however, the proposed schemes achieve significant AoI reductions compared to the Random Offloading strategy. Nevertheless, as the scale of IIoT devices gradually increases (e.g., $N = 16$), the proposed algorithm achieves AoI reductions of 4\%, 14.7\%, 18.7\%, and 22.8\% over D3QN, DQN, Greedy, and Random Offloading, respectively. These results demonstrate that the proposed algorithm’s performance advantage increases with the scale of IIoT devices, underscoring its scalability and efficacy in dense network environments.

\begin{figure}[htbp]
\centering
\includegraphics[scale=0.4]{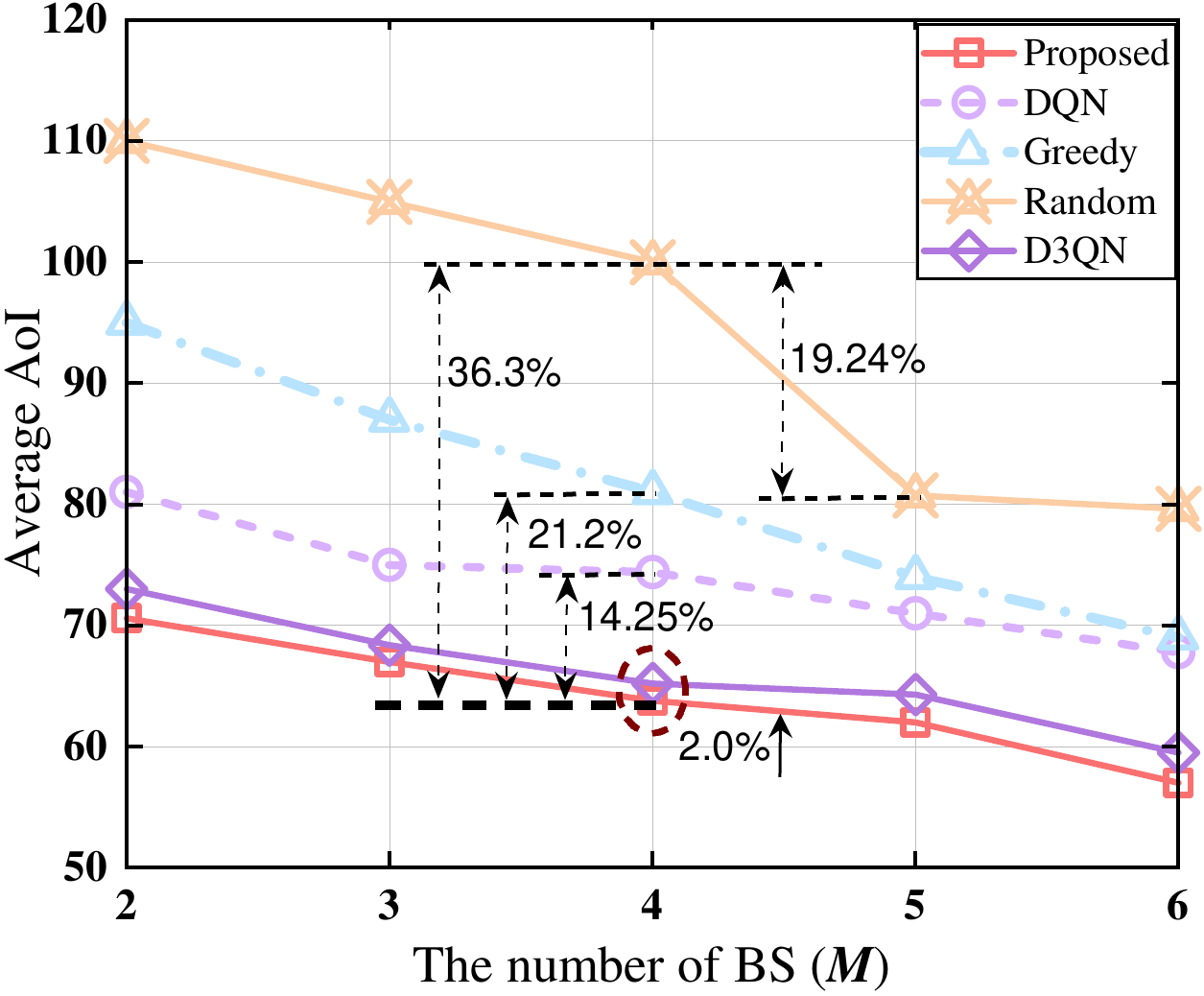}
\caption{\small Performance comparison of average AoI with different numbers of BS-MEC servers.}
\label{fig_AAoI_vs_BS}
\end{figure}

\par As depicted in Fig. \ref{fig_AAoI_vs_BS}, we analyze the system's average AoI performance under different numbers of BSs. It can be seen that as the number of BSs increases, the average AoI performance is significantly improved. For instance, with the Random Offloading strategy, when the number of BSs increases from 4 to 5, the system's average AoI performance improves by 19.24\%. However, the continuous increase in the number of BSs does not significantly improve the performance of the average AoI. This is primarily because when the number of BSs exceeds a certain threshold, the impact of performance degradation caused by the uneven offloading resulting from the random deployment of the Random Offloading strategy diminishes. Additionally, when $M = 4$, the proposed algorithm shows a significant improvement in average AoI performance compared to D3QN, DQN, and Greedy, with performance gains of 2\%, 14.25\%, 21.2\%, and 36.3\%, respectively. These results highlight the proposed algorithm's superior scalability and efficacy in leveraging increased BS availability to optimize data freshness in multi-BS IIoT systems.

\par Drawing from Fig. \ref{fig_AAoI_vs_IIoT} and Fig. \ref{fig_AAoI_vs_BS}, the superior performance of the proposed schemes stems from their efficient learning and selection of optimal task offloading strategies, significantly reducing the system’s average AoI. For example, when the number of IIoT devices is $N = 10$ and the number of base stations is $M = 5$, the proposed schemes reduce the average AoI by 23\% and 12\% compared to the Greedy algorithm and DQN algorithm, respectively, and outperform the D3QN algorithm. These results suggest that the introduced branching structure can effectively enhance the system's real-time performance and data update efficiency by quickly determining the best offloading decisions. Furthermore, we can observe that increasing the number of IIoT devices can heighten the competition for offloading among IIoT devices, and the contention for BS services becomes more intense, leading to an increase in the system's average AoI. Conversely, more BSs reduce resource contention since more IIoT devices have the opportunity to offload data to the BSs, thereby reducing the average AoI. These dynamics highlight the need to balance device and BS counts in offloading strategies. While additional BSs enhance data freshness, scaling IIoT devices increases costs. Therefore, in practical applications, appropriate algorithms and strategies should be adopted based on specific requirements and resource conditions to achieve optimal system performance.

\begin{figure}[htbp]
\centering
\includegraphics[scale=0.4]{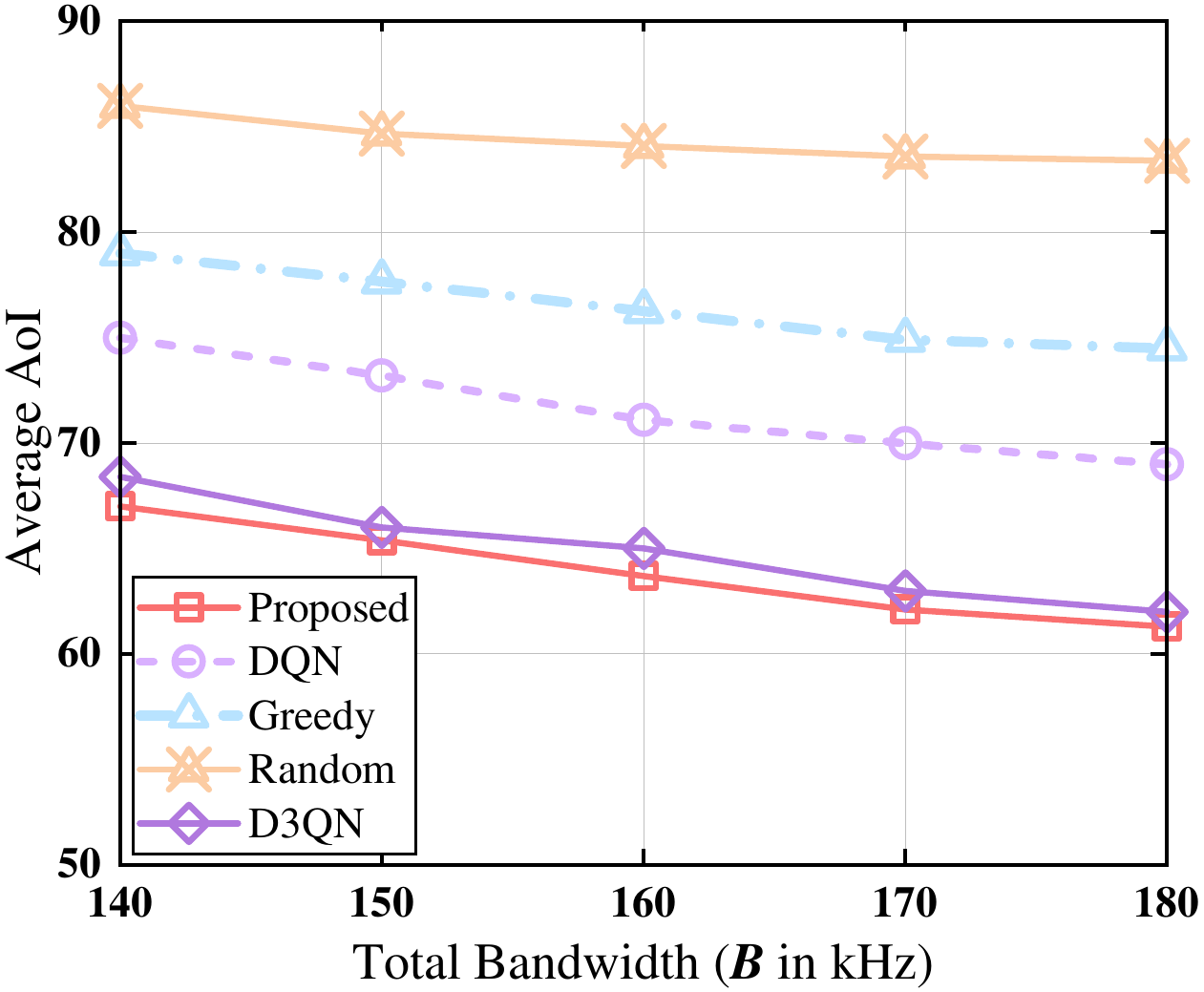}
\caption{\small Performance comparison of average AoI with varying bandwidth resources.}
\label{fig_AAoI_vs_BW}
\end{figure}

\par Additionally, resource allocation also significantly influences the system's AoI performance. As shown in Fig. \ref{fig_AAoI_vs_BW} and Fig. \ref{fig_AAoI_vs_Comp}, as bandwidth and computational resources increase, the system's average AoI decreases, indicating an improvement in the system's timeliness. According to Eq. (\ref{eq3}), when the amount of bandwidth allocated to a device increases, the device's transmission time will decrease. Similarly, as indicated by Eq. (\ref{eq5}), when the amount of computational resources obtained by a device increases, its computation time will also shorten accordingly. These two changes are directly related to the calculation of the system's AoI. Based on the definition of AoI in Eq. (\ref{eq8}), the system's average is closely related to the transmission and computation delays. However, as can be further observed from Fig. \ref{fig_AAoI_vs_BW} and Fig. \ref{fig_AAoI_vs_Comp}, when bandwidth and computational resources increase to a certain threshold, the AoI reduction of IIoT devices gradually slows down and even tends to stabilize. This is mainly because excessive resource allocation yields marginal improvements in transmission and computation delays, limiting further AoI decreases. Therefore, optimal resource allocation is critical to balance timeliness and resource utilization in real-time systems.

\begin{figure}[htbp]
\centering
\includegraphics[scale=0.4]{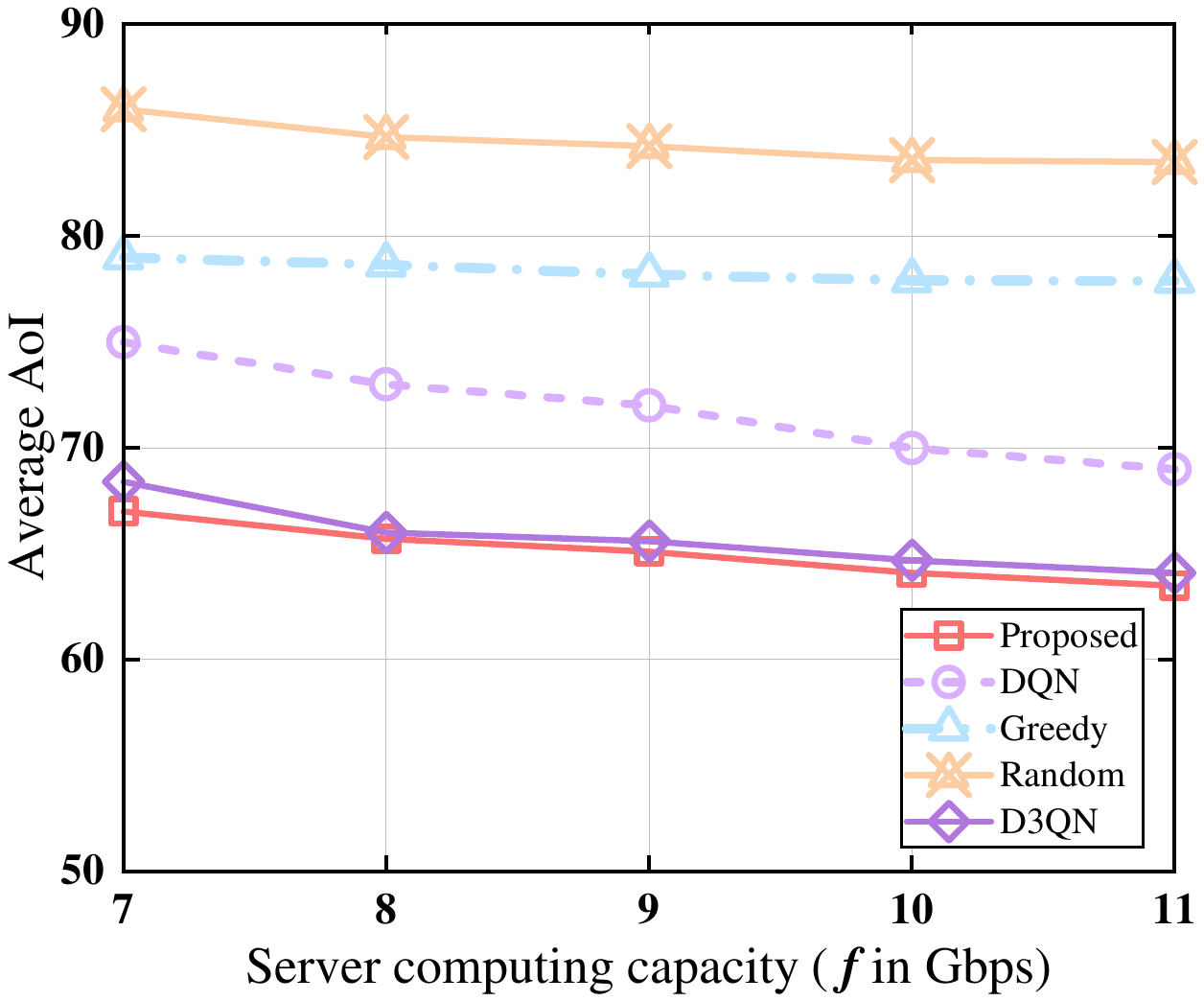}
\caption{\small Performance comparison of average AoI with varying server computation resources.}
\label{fig_AAoI_vs_Comp}
\end{figure}

\section{Conclusion}
In this paper, we propose an AoI-aware multi-BS-MEC real-time monitoring system to facilitate IIoT access, where a joint task offloading and resource allocation optimization problem is formulated to minimize the long-term average AoI. To effectively tackle this non-convex dynamic stochastic optimization problem, we decomposed it into the task offloading and resource allocation subproblems. The task offloading problem is equivalently reformulated as a CMDP and solved via a novel Branching D3QN algorithm, which reduces the action-space complexity from exponential to linear levels. The resource allocation subproblem is proved convex by analyzing the semi-definite Hessian matrix of bandwidth and computation resources. By integrating both via an alternating optimization framework, we develop a joint scheme that achieves near-optimal AoI performance. Extensive simulations demonstrate that the proposed BD3QN algorithm surpasses state-of-the-art DRL and heuristic methods, achieving up to 75\% faster convergence and at least 22\% lower long-term average AoI.

\par Although this work addresses large-scale IIoT access in multi-BS scenarios, it assumes that each IIoT device is associated with a single BS. In real-world scenarios, devices may dynamically switch among multiple BS, leading to service migration. In future work, we intend to extend our framework to incorporate service migration, aiming to design more robust task offloading strategies with stronger adaptability to highly dynamic IIoT environments.

\footnotesize
\bibliographystyle{IEEEtranN}
\bibliography{IEEEabrv,ref}

\vspace{-0.8cm}
\begin{IEEEbiography}[{\includegraphics[width=1in,height=1.25in,clip,keepaspectratio]{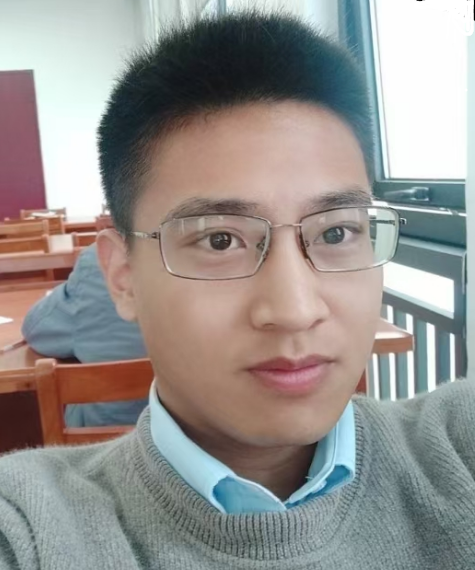}}]{Yuang Chen (Graduate Student Member, IEEE)}
	received the B.S. degree from the Hefei University of Technology (HFUT), Hefei, China, in 2021. He is currently pursuing a doctoral degree in Electronic Engineering and Information Science at the University of Science and Technology of China, Hefei, China. In addition, he is currently working as a research assistant in the department of computing at The Hong Kong Polytechnic University, Hong Kong SAR, China. His research interests include 5G/6G wireless network technology, such as next-generation URLLC, next-generation multiple access technology, wireless network resource allocation and performance optimization, microservice deployment and scheduling, etc. 
\end{IEEEbiography}

\vspace{-1cm}

\begin{IEEEbiography}[{\includegraphics[width=1in,height=1.25in,clip,keepaspectratio]{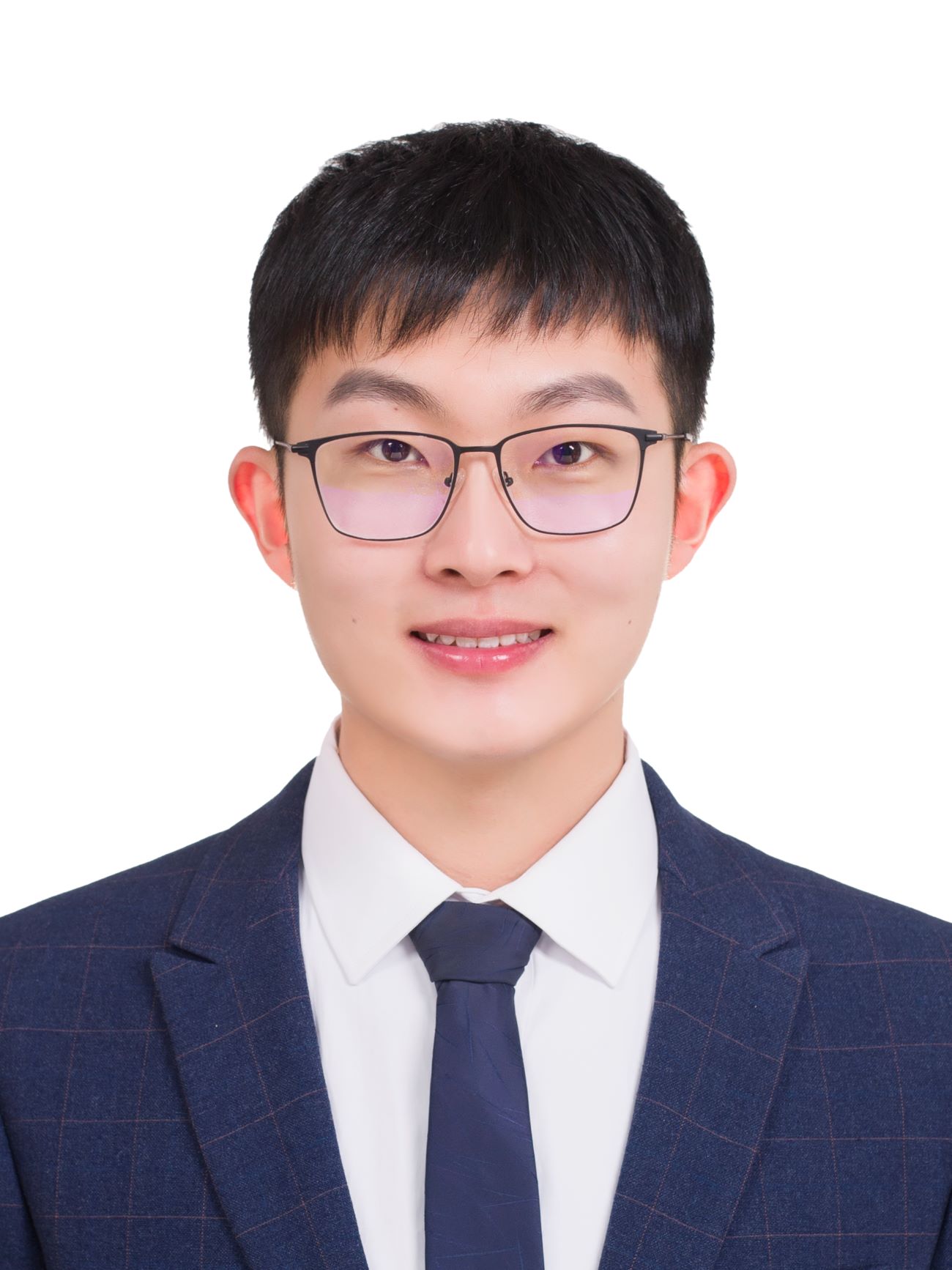}}]{Fengqian Guo}
	received the Ph.D. degree in communication and information systems from the University of Science and Technology of China (USTC), Hefei, China, in 2022. He is currently a jointly trained postdoctoral Researcher with the University of Science and Technology of China and Tencent. His research interests include wireless low-latency transmission and wireless resource optimization.
\end{IEEEbiography}

\vspace{-1cm}

\begin{IEEEbiography}[{\includegraphics[width=1in,height=1.25in,clip,keepaspectratio]{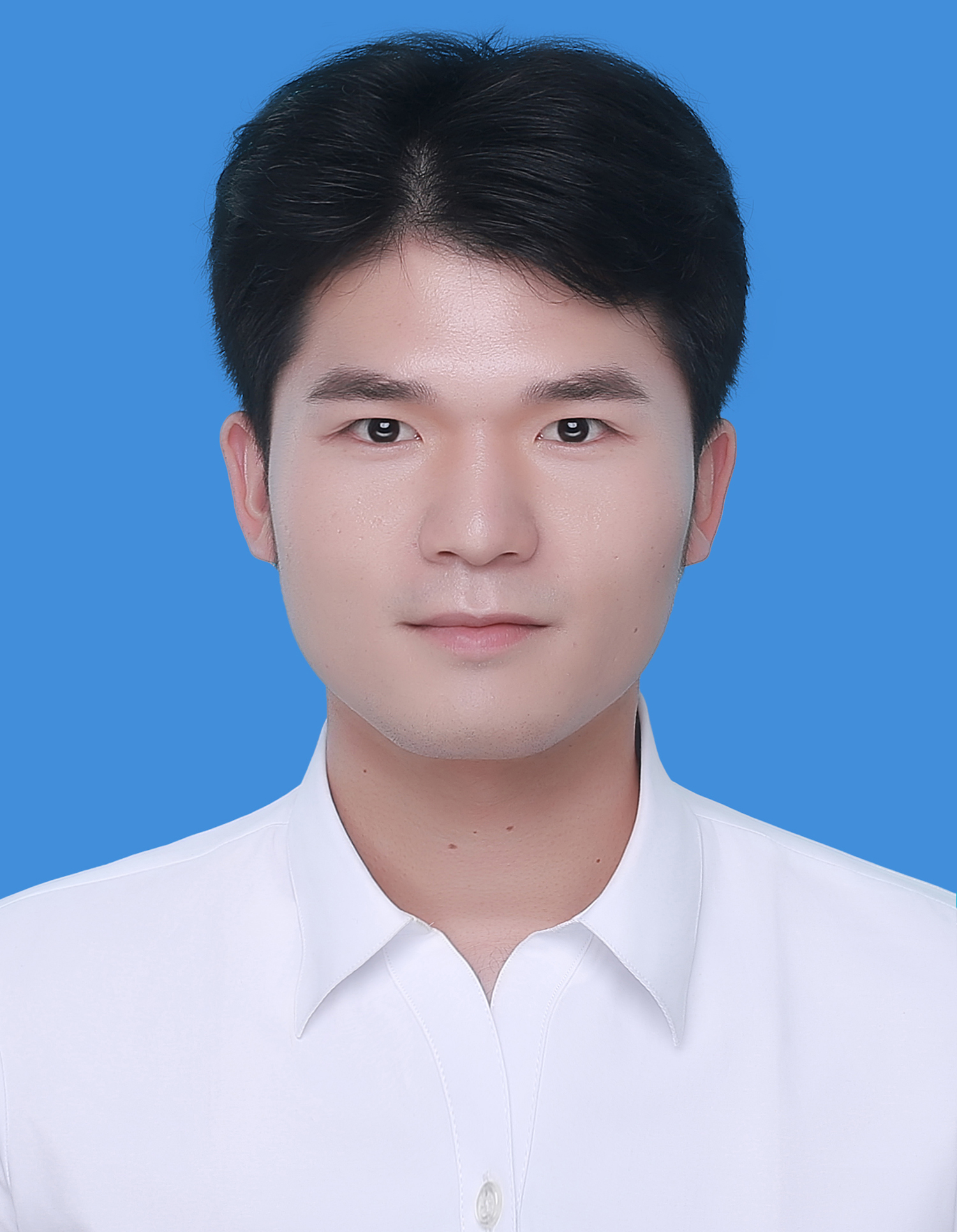}}]{Chang Wu}
	received the B.S. degree from the Dalian Maritime University (DLMU), Dalian, China, in 2021. He is currently working toward the PhD degree in communication and information systems with the Department of Electronic Engineering and Information Science, University of Science and Technology of China (USTC), Hefei, China. His research interests include 5G/6G wireless network technologies, such as architecture, QoS/QoE provision for business transmission and Deep Reinforcement Learning in performance optimization.
\end{IEEEbiography}

\vspace{-1cm}

\begin{IEEEbiography}[{\includegraphics[width=1in,height=1.2in,clip,keepaspectratio]{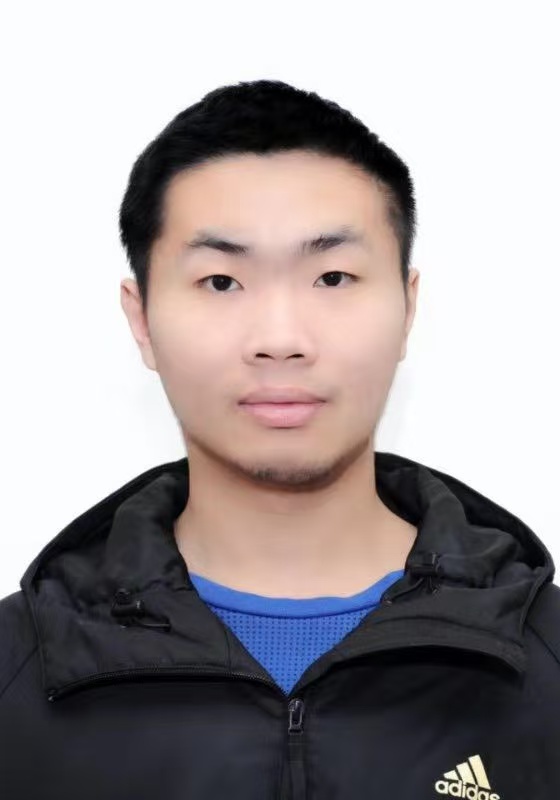}}]{Shuyi Liu}
	received the B.S. degree from Lanzhou University (LZU), Lanzhou, China, in 2021. He iscurrently pursuing the Ph.D.degree with the Departmentof Electronic Engineering and Information Science, University of Science and Technology ofChina (USTC), Hefei, China. His research interests include traffic engineering, segment routing, and deep reinforcement learning.
\end{IEEEbiography}

\vspace{-1.3cm}

%
%
%
%
%
\begin{IEEEbiography}[{\includegraphics[width=1in,height=1.25in,clip,keepaspectratio]{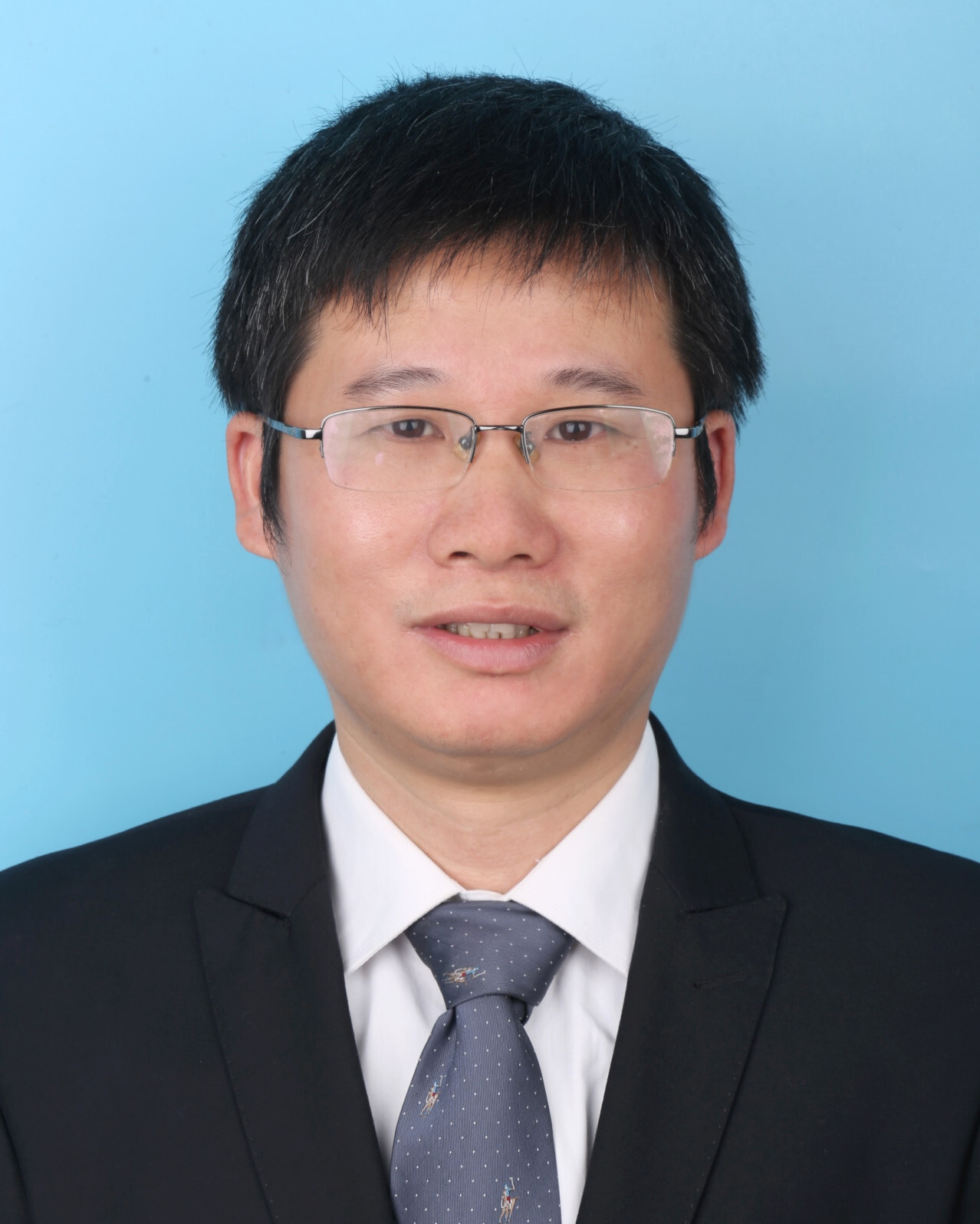}}]{Hancheng Lu (Senior Member, IEEE)} received his Ph.D. in communication and information systems from the University of Science and Technology of China, Hefei, China, in 2005. He is currently a tenured professor in the Department of Electronic Engineering and Information Science at the University of Science and Technology of China. He is also working at the Hefei National Comprehensive Science Center Artificial Intelligence Research Institute, Hefei, China. He has rich research experience in multimedia communication, wireless edge networks, future network architecture and protocols, as well as machine learning algorithms for network communication, involving scheduling, resource management, routing, transmission, and other fields. In the past 5 years, more than 80 papers have been published in top journals such as IEEE Trans and flagship conferences such as IEEE INFOCOM in this field, and have won the Best Paper Award of IEEE GLOBECOM 2021 and the Best Paper Award of WCSP 2019 and WCSP 2016 in the field of communication. In addition, he currently serves as an editorial board member for numerous journals such as the IEEE Internet of Things Journal, China Communications, and IET Communications.
\end{IEEEbiography}

\vspace{-1cm}

\begin{IEEEbiography}[{\includegraphics[width=1.4in,height=1.3in,clip,keepaspectratio]{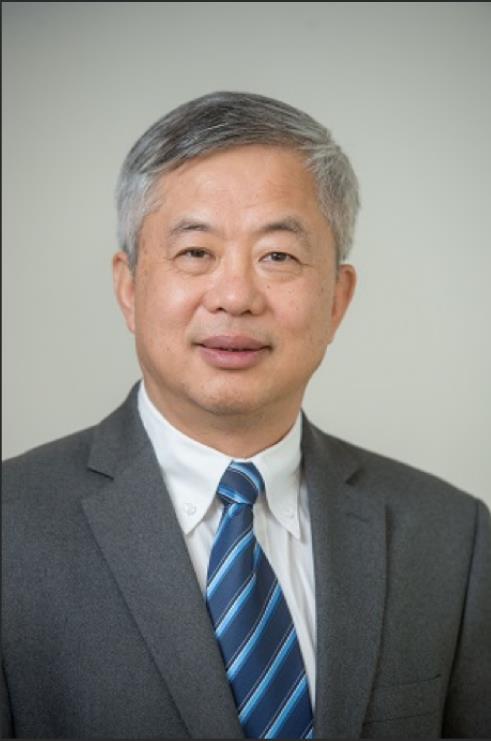}}]{Chang Wen Chen (Life Fellow, IEEE)} received the B.S. degree from the University of Science and Technology of China in 1983, the M.S.E.E. degree from the University of Southern California in 1986, and the Ph.D. degree from the University of Illinois at Urbana–Champaign in 1992. He is currently the Chair Professor in the Department of Visual Computing and the Interim Dean of the School of Computer and Mathematical Sciences at The Hong Kong Polytechnic University. Before his current position, he served as the Dean of the School of Science and Engineering, The Chinese University of Hong Kong, Shenzhen, from 2017 to 2020. He was an Empire Innovation Professor with the University at Buffalo, The State University of New York, from 2008 to 2021. He was an Allen Henry Endow Chair Professor with Florida Institute of Technology from 2003 to 2007. He was a Faculty Member of electrical and computer engineering at the University of Rochester from 1992 to 1996 and at the University of Missouri-Columbia from 1996 to 2003. His research has been funded by both government agencies and industrial corporations. His research interests include multimedia communication, multimedia systems, mobile video streaming, the Internet of Video Things (IoVT), image/video processing, computer vision, deep learning, multimedia signal processing, and immersive mobile video. He was a SPIE Fellow in 2007 and an Elected Member of Academia Europaea in 2021. He and his students have received ten best paper awards or best student paper awards over the past two decades. He received several research and professional achievement awards, such as the Sigma Xi Excellence in Graduate Research Mentoring Award in 2003, the Alexander von Humboldt Research Award in 2010, the University at Buffalo Exceptional Scholar—Sustained Achievement Award in 2012, the SUNY System Chancellor's Award for Excellence in Scholarship and Creative Activities in 2016, the University of Illinois ECE Distinguished Alumni Award in 2019, the Outstanding Overseas Contributor of the China Society of Image and Graphics (CSIS) in China MM 2024, and the SIGMM Outstanding Technical Achievement Award in 2024. He is currently an Associate Editor-in-Chief of IEEE TRANSACTIONS ON BIOMETRICS, BEHAVIOR, AND IDENTITY SCIENCE and a Deputy Editor-in-Chief of the IEEE TRANSACTIONS ON IMAGE PROCESSING. He served as the conference chair for several major IEEE, ACM, and SPIE conferences related to multimedia communications and signal processing. He served as the Editor-in-Chief for IEEE TRANSACTIONS ON MULTIMEDIA from January 2014 to December 2016 and IEEE TRANSACTIONS ON CIRCUITS AND SYSTEMS FOR VIDEO TECHNOLOGY from January 2006 to December 2009. He has been an Editor of several other major IEEE TRANSACTIONS and journals, including PROCEEDINGS OF THE IEEE, IEEE JOURNAL OF SELECTED TOPICS IN SIGNAL PROCESSING, IEEE JOURNAL OF SELECTED AREAS IN COMMUNICATIONS, and IEEE JOURNAL OF SELECTED TOPICS IN SIGNAL PROCESSING.
\end{IEEEbiography}

\end{document}